\documentclass[reqno]{amsart}

\usepackage{amssymb}
\usepackage{graphicx}
\usepackage{enumerate}

\usepackage{color}

\usepackage{slashed,enumerate}

\usepackage[utf8]{inputenc}

\title{Localization for transversally periodic random potentials on binary trees}

\author{Richard Froese}
\author{Darrick Lee}
\author{Christian Sadel}
\author{Wolfgang Spitzer}
\author{G\"unter Stolz}

\address[Froese]{Department of Mathematics, University of British Columbia, 1984 Mathematics Road, Vancouver, BC, \mbox{V6T\,1Z2}, Canada}
\email{rfroese@math.ubc.ca}
\address[Lee]{Department of Mathematics, University of British Columbia, 1984 Mathematics Road, Vancouver, BC, \mbox{V6T\,1Z2}, Canada}
\email{darrick.l@hotmail.com}
\address[Sadel]{Department of Mathematics, University of British Columbia, 1984 Mathematics Road, Vancouver, BC, \mbox{V6T\,1Z2}, Canada; 
Institute of Science and Technology, Am Campus 1, 3400 Klosterneuburg, Austria}
\email{christian.sadel@ist.ac.at}
\address[Spitzer]{Fakult\"at f\"ur Mathematik und Informatik, Fernuniversit\"at Hagen, Universit\"atsstras\ss e 1, 58097 Hagen, Germany}
\email{wolfgang.spitzer@fernuni-hagen.de}

\address[Stolz]{Department of Mathematics, University of Alabama at Birmingham, Birmingham, AL 35294-1170, USA} \email{stolz@uab.edu}



\newtheorem{theorem}{Theorem}
\newtheorem{lemma}[theorem]{Lemma}
\newtheorem{coro}[theorem]{Corollary}
\newtheorem{prop}[theorem]{Proposition}
\newtheorem{remark}[theorem]{Remark} 
\newtheorem{defini}[theorem]{Definition}

\newtheorem{assumptions}{Assumptions}

\theoremstyle{remark}
\newtheorem*{rem}{Remark}

\newcommand{\Aa}{{\mathcal A}}

\newcommand{\Hh}{{\mathcal H}}

\newcommand{\BB}{{\mathbb B}}
\newcommand{\CC}{{\mathbb C}}

\newcommand{\EE}{{\mathbb E}}

\newcommand{\HH}{{\mathbb H}}

\newcommand{\NN}{{\mathbb N}}

\newcommand{\PP}{{\mathbb P}}

\newcommand{\RR}{{\mathbb R}}

\newcommand{\TT}{{\mathbb T}}

\newcommand{\ZZ}{{\mathbb Z}}

\newcommand{\one}{{\bf 1}}
\newcommand{\nul}{{\bf 0}}

\newcommand{\qtx}[1]{\quad\text{#1}\quad}

\newcommand{\pmat}[1]{\begin{pmatrix} #1  \end{pmatrix}}

\DeclareMathOperator{\re}{{\rm Re}}

\newcommand{\wlim}{\mathop{\hbox{\rm w-lim}}}

\newcommand{\be}[1]{\begin{equation} \label{#1} }
\newcommand{\ee}{\end{equation}}
\newcommand{\beq}{\begin{equation}}
\newcommand{\rf}[1]{\eqref{#1}}

\renewcommand\Im{\mathop{\rm Im}}
\renewcommand\Re{\mathop{\rm Re}}

\newcommand\arcosh{\mathop{\rm arcosh}}

\numberwithin{equation}{section}
\numberwithin{theorem}{section}

\begin{document}
 
\begin{abstract}
We consider a random Schrödinger operator on the binary tree with a random potential which is the sum of a random radially symmetric potential, $Q_r$, and a random transversally periodic potential, $\kappa Q_t$, with coupling constant $\kappa$. Using a new one-dimensional dynamical systems approach combined with Jensen's inequality in hyperbolic space (our key estimate) we obtain a fractional moment estimate proving localization for small and large $\kappa$. Together with a previous result we therefore obtain a model with two Anderson transitions, from localization to delocalization and back to localization, when increasing $\kappa$. As a by-product we also have a partially new proof of one-dimensional Anderson localization at any disorder.
\end{abstract}
 
 \maketitle

\section{Introduction and statement of results}

In this work we consider discrete Schr\"odinger operators of the form $H=-\Delta + Q$ acting in $\ell^2(\BB)$, where $\BB$ is a rooted regular tree (or Bethe Lattice), $\Delta$ is the adjacency operator and $Q$ is a bounded random potential. 

For the Anderson model, where the values of $Q$ are independent and identically distributed, the spectrum of $H$ may have an absolutely continuous component. In fact,
if the tree has connectivity $k+1$ and the single site distribution has support $[-\kappa,\kappa]$,
then $\sigma(H)=[-2\sqrt{k}-\kappa, 2\sqrt{k}+\kappa]$ almost surely, and if $k\ge 2$ it is known from the work of Klein \cite{K} and Aizenman-Warzel \cite{AW1,AW2} that for small $\kappa$ and suitable single site distribution, the spectrum is purely absolutely continuous almost surely inside $[-2\sqrt{k}, 2\sqrt{k}]$ and near the endpoints $\pm (2\sqrt{k}+\kappa)$.

In fact, for the Anderson model on graphs (adjacency operator plus random i.i.d. potential) the existence of absolutely continuous spectrum at low disorder has only been shown for trees and tree-like graphs\footnote{e.g. adding some loops to trees or taking cross products of trees with finite graphs} of infinite dimension with exponentially growing boundary. A lot of work has been done in extending Klein's original result, also in recent years, \cite{ASW, AW1, AW2, FHS1, FHS1b, FHS2, FHH, Hal, KLW, KlS, Sa, Sh}.
At large disorder and on the edge of the spectrum one typically finds Anderson localization (pure point spectrum) in any dimension \cite{A, AM, CKM, DK, DLS, FMSS, FS, Klo, W}. 
Pure point spectrum for small disorder typically appears in one and quasi one-dimensional models \cite{CKM, GMP, KLS, Lac}, unless localization is prevented by some symmetry\footnote{such symmetries appear in effective models for topological insulators} \cite{SS}.

Now, if $Q$ is a radial potential, where the common potential values for each level are independent and identically distributed, then $H$ can be decomposed as a direct sum of one-dimensional Anderson Hamiltonians. So in this case there is localization at all non-zero values of $\kappa$, and $\sigma(H)=\sigma_{\rm{pp}}(H)=[-2\sqrt{k}-\kappa, 2\sqrt{k}+\kappa]$ almost surely. % (see \cite{HP}).

In this paper, we consider the class of transversally periodic potentials. These potentials were introduced in \cite{FHS1} and \cite{FHS2} to illustrate how fast transversal oscillations can generate absolutely continuous spectrum in a strongly correlated model. However, as this class includes radial potentials as a special case, absolutely continuous spectrum need not be present. We show that the pure point spectrum of a radial potential is stable under small transversally periodic perturbations that destroy the radial symmetry. For simpler reasons, we can also show that the spectrum is pure point for large transversally periodic perturbations. Using results of \cite{FHS2} we can therefore construct an example where there are two 
Anderson transitions when increasing the disorder of an added, independent non-radial, transversally periodic potential.

Now let us consider the geometry of the problem more precisely.
The graph distance on $\BB$ will be denoted by $d(x,y)$ for vertices $x,y\in \BB$.
The root will be denoted by $0$. Although we believe our methods can handle any $k$, we set $k=2$ in this paper for the sake of simplicity. Thus the root has two neighbors and any other vertex has three neighbors.
The $n$th sphere for $n\in\NN_0$ is denoted by $S_n:=\{x\in \BB\,:\,d(0,x)=n\}$.

\begin{center}
 \includegraphics[width=5cm]{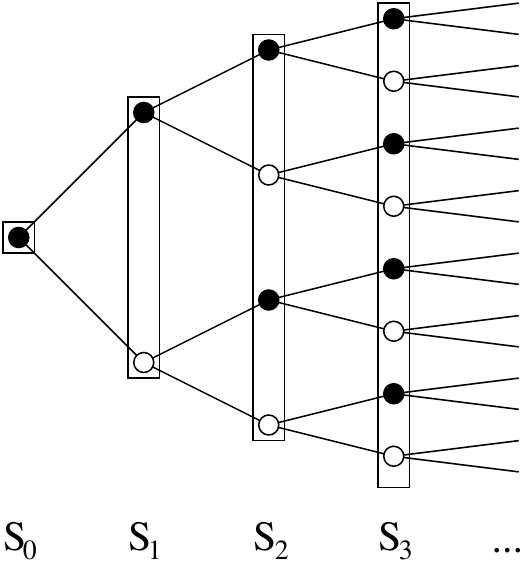}
\end{center}

For $\kappa\ge 0$, we define the operator $H_\kappa$ on $\ell^2(\BB)$ by
\be{model}
H_{\kappa} = -\Delta + Q_r + \kappa Q_t\qtx{i.e.}
(H_\kappa \psi)(x) = -\!\!\!\!\!\!\!\sum_{y:d(x,y)=1}\!\!\!\!\! \psi(y) + Q_r(x)\psi(x) + \kappa Q_t(x) \psi(x)\;.
\ee
Here, $Q_r:\BB\to\RR$ is a radial potential, i.e., for any $x\in S_n$ one has $Q_r(x)=q_n$. The $q_n$ are chosen independently, according to a distribution $\nu_0$ for $n=0$ and identically according to $\nu$ for $n\ge 1$. 
The transversally periodic random potential $Q_t:\BB\to \RR$, coupled with the coupling constant $\kappa$, is independent of $Q_r$ and defined as follows. We choose $Q_t(0)=0$ and for each level or sphere in the tree  except the first, there are two potential values chosen at random. These values are then repeated periodically across the level. Thus on the diagram, for each level, all the black vertices have a common value, as do all the white ones. Each pair of potentials is chosen independently from a common joint distribution $\sigma$. 

The purpose of this work is to show localization for small and large transversally periodic disorder $\kappa$ for any disorder in the radial potential. Our main results are for small $\kappa$. For these we will need the following assumptions.

\begin{assumptions}\label{assumptions1} The radial measures $\nu_0$ and $\nu$ have bounded densities (denoted also by $\nu_0$ and $\nu$) with support in $[-K,K]\subset\RR$ for some $K<\infty$. The transversal measure $\sigma$ has support in $[-1,1]^2\subset\RR^2$.
\end{assumptions}  

\begin{theorem}\label{smallkappathm} Let $H_{\kappa} = -\Delta + Q_r + \kappa Q_t$ where the distributions for the radial potential $Q_r$ and the transversally periodic potential $Q_t$ satisfy Assumptions \ref{assumptions1}. Then there exists a $\kappa_0>0$ such that for $0\le \kappa \le \kappa_0$, $H_{\kappa}$ has pure point spectrum at all energies almost surely.
\end{theorem}

We will prove localization via the fractional moment estimate of Aizenman and Molchanov \cite{A,AM,H,Sto} and the Simon-Wolff criterion \cite{SW}. The required version of the fractional moment estimate, contained in the following theorem, is our main technical result. Its proof introduces a new dynamical systems approach.

\begin{remark}\label{rm:1D}
 Note that Theorem~\ref{smallkappathm} specifically includes the case $\kappa=0$, this means that we also give an alternative proof for localization for random radial potentials.
 This part works for any connectivity $k+1$, even for the case $k=1$, giving an alternative proof for Anderson localization in one dimension (see also Remark~\ref{rm:1D-proof} stating the 1D proof).
\end{remark}

We use the notation $|x\rangle\langle y|$ for the rank one operator $\delta_x\otimes\delta_y^*$, and define the projections $P_B := \sum_{x\in B}|x\rangle\langle x|$ for $B\subseteq \BB$ and $P_n := P_{S_n}$.

\begin{theorem}\label{mainthm} Let $H_{\kappa} = -\Delta + Q_r + \kappa Q_t$ where the distributions for the radial potential $Q_r$ and the transversally periodic potential $Q_t$ satisfy Assumptions \ref{assumptions1}. Given any energy $E_0\in\RR$, there exist an open interval $I\subset \RR$ containing $E_0$, constants $s\in (0,1)$, $\epsilon_0>0$, $\kappa_0>0$, $C<\infty$ and $\ell>1$ such that for all $\kappa\in [0,\kappa_0]$ and $n \ge 0$,
\begin{equation}\label{fracmoment}
\sup_{E\in I}\sup_{0<\epsilon<\epsilon_0}\EE\left[\|P_0 (H_\kappa-E-i\epsilon)^{-1} P_n\|^s\right] \le C \ell^{-n}.
\end{equation}
\end{theorem}

For the large $\kappa$ result we need the following assumptions:

\begin{assumptions}\label{assumptions2} In addition to Assumptions \ref {assumptions1}, assume that the marginal measures $\sigma_0$ and $\sigma_1$ of $\sigma$ each have bounded densities.
\end{assumptions}

\begin{theorem}\label{largekappathm} Let $H_{\kappa} = -\Delta + Q_r + \kappa Q_t$ where the distributions for the radial potential $Q_r$ and the transversally periodic potential $Q_t$ satisfy Assumptions \ref{assumptions2}. Then there exists a $\kappa_1<\infty$ such that for $\kappa \ge \kappa_1$, $H_{\kappa}$ has pure point spectrum at all energies almost surely.
\end{theorem}

It follows from the results of \cite{FHS2} that there are potentials of the form $Q_r + Q_t$ with an interval of absolutely continuous spectrum. Starting with such a potential and introducing a coupling constant $\kappa$, we obtain the following corollary.
\begin{coro}\label{doubletransition} There exist a random radial potential $Q_r$ and a random transversally periodic potential $Q_t$ such that $H_\kappa=-\Delta + Q_r + \kappa Q_t$ has pure point spectrum for large and small $\kappa$ and an interval of purely absolutely continuous spectrum for some intermediate values of $\kappa$. 
\end{coro}

Here is an outline of our paper. In Section \ref{sectionproof} we show how Theorem \ref{smallkappathm} follows from the fractional moment estimate in Theorem \ref{mainthm}. In Section \ref{sectionreduction} we express the fractional moment estimate in terms of a dynamical system. We prove continuity at $\kappa=0$ and reduce considerations to a one-dimensional system
(cf. Remark~\ref{rm:1D}). Section~\ref{sectionkappazero} contains the proof for this one-dimensional system, modulo a key estimate. As will be explained in a remark, this also provides an alternative proof of one-dimensional (dynamical) localization, first established by Kunz and Soulliard under the assumptions used here (\cite{KS}, see also \cite{S}). The key estimate is proved in Section \ref{sectionkey}. Finally, Section~\ref{sectionproof2} contains the proofs of Theorem~\ref{mainthm} and Theorem~\ref{largekappathm}. 

In the one-dimensional setting, our dynamical systems approach is based on the analysis of the operator $T_{E,s}:C(\overline{\RR})\rightarrow C(\overline{\RR})$ defined by
\beq \label{1Ddynsyst}
(T_{E,s}f)(w) = \mathbb E_q\left[f(-1/(w+E-q))\,|w+E-q|^{-s}\right].
\ee
The fractional moment estimate for energies near $E$, which implies localization, follows from 
the estimate
$\|T_{E,s}^m 1\|_\infty$ $ < 1$ for some $m\in\mathbb N$. The observation that this is true for $m=1$ if the distribution for the random potential $q$ is sufficiently spread out is essentially the Aizenman-Molchanov proof of large disorder localization. In order to obtain this estimate for more general distributions, we show that there exist special bounding functions which decrease point-wise under an application of $T$. Any positive function lying below such a bounding function is then forced to zero under repeated applications of $T$. This forms the basis of our proof.

We conclude the introduction with some open problems.
\begin{enumerate}[(i)]
\item It would be interesting to know the spectral properties of $ -\Delta + Q_r + \kappa Q_A$ for small $\kappa$, where $Q_A$ is an Anderson (independent, identically distributed) potential on $\BB$. More generally, one might ask if Theorem~\ref{smallkappathm} can be generalized to say that $-\Delta + Q_r + Q$ has pure point spectrum for every deterministic perturbation $Q$ of $-\Delta +Q_r$ with sufficiently small $L^{\infty}$-norm. The dynamical systems approach which we apply below to the case of transversally periodic potentials does not work for any of these settings.

\item The Simon-Wolff argument we use implies spectral localization. Dynamical localization in the form $\EE[\sup_{t\in \RR} \|P_m e^{-itH} P_n\|] \le C \ell^{-|m-n|}$ would follow from an estimate $\EE[\|P_m (H-z)^{-1}P_n\|^s]\le C \ell^{-|m-n|}$ (by modifying the proof in \cite{Sto}). However, we are not able to prove this unless $m=0$ or $n=0$.

\item For the intermediate values of $\kappa$, where the potentials in Corollary \ref{doubletransition} exhibit an interval of absolutely continuous spectrum, it would be interesting to know whether there is band edge localization.  
\end{enumerate}

\section{Proof of Theorem \ref{smallkappathm}}\label{sectionproof}

Recall that the random potential $Q = Q_r +\kappa Q_t$ is defined by choosing a value at the root according to $\nu_0$ and a pair of values for each sphere $S_n$ with $n\ge1$. The radial component of the pair is chosen according to the push-forward $D_*\nu$, where $D:\RR\to\RR^2$ defined as $D(x) := (x,x)$ is the diagonal map and $\nu$ is the radial measure for spheres $S_n$ with $n\ge1$. The perturbation is chosen according to the scaled distribution $\sigma_\kappa$ defined by $\sigma_\kappa(A):=\sigma(A/\kappa)$ for measurable sets $A\subseteq\RR^2$. The distribution for the sum is the convolution $\mu :=( D_*\nu)\star\sigma_\kappa$, so that integration on $\RR^2$ with respect to $\mu$ is characterized as
\beq
\int_{\RR^2} g(q)\,d\mu(q) = \int_{[-1,1]^2} \int_{\RR} g(r+\kappa p_0, r+\kappa p_1)\,\nu(r)\,dr\,d\sigma(p_0,p_1)
\ee
for suitable functions $g$.

Our measure space is $\Omega:=\RR\otimes\bigotimes_{n\ge1} \RR^2$ with the product measure $\mathbb P:=\nu_0\otimes\bigotimes_{n\ge1} \mu$. We define $Q$ by repeating the pair of values for each $n\ge1$ periodically across $S_n$. The distribution of the resulting random variable $\Omega\ni\omega\mapsto Q(\omega)$ defines a measure on $\BB^\RR$. When restricting to forward 
trees, we introduce the measures $\mathbb P_0 := \mu_0\otimes\bigotimes_{n\ge1} \mu$ and $\mathbb P_1 := \mu_1\otimes\bigotimes_{n\ge1} \mu$ on $\RR\otimes\bigotimes_{n\ge1} \RR^2$, where $\mu_0$ and $\mu_1$ denote the two marginal probability measures of $\mu$ on $\RR$, given for suitable $f$ by
\beq
\int_{\RR} f(q_i)\, d\mu_i(q_i) = \int_{\RR^2} f(q_i)\,d\mu(q_0,q_1), \quad i=0,1,
\ee
and which take on the role of the measure $\nu_0$ at the new root. One easily sees that  $\mu_i$, $i=0,1$, are absolutely continuous with densities
\beq \label{margdensities}
\int_{\RR} \nu(q_i-\kappa p_i)\,d \sigma_i(p_i),
\ee
as functions of $q_i$, $i=0,1$. Here $\sigma_i$ denote the marginal measures of $\sigma$. Thus the densities (\ref{margdensities}) are bounded by $\|\nu\|_{\infty}$ and have compact support in $[-K-\kappa, K+\kappa]$.

For $x\in\BB$, let $\pi_x:\Omega\to\Omega$ denote the map
\beq
\pi_x\big(q_0;(q_{10},q_{11});\ldots;(q_{d(x)0},q_{d(x)1});\ldots\big) := \big(q_{d(x)c(x)};(q_{(d(x)+1)0},q_{(d(x)+1)1});\ldots\big),
\ee
where $d(x):=d(x,0)$ and $c(x) = 0$ (respectively $1$) if $x$ is a black (respectively white) vertex. Then $(\pi_x)_*\mathbb P = \mathbb P_{c(x)}$ by the definition of marginals. Notice that if $A$ corresponds to a set of potentials on a subtree $\BB_x$ with $c(x) = i$ and $\mathbb P_i(A) = 1$, then $\pi_x^{-1}(A)$ corresponds to potentials on the original tree whose restrictions to $\BB_x$ lie in $A$, and $\mathbb P\big[\pi_x^{-1}(A)\big] = \big((\pi_x)_*\mathbb P\big)(A) = 1$.

\medskip
We prove Theorem \ref{smallkappathm} in two steps. Firstly, using Theorem \ref{mainthm} and the Simon-Wolff criterion we show that $H_\kappa$ has pure point spectrum at the root $0$ of the full tree $\BB$. This can be applied to the restriction of the Hamiltonian $H_\kappa$ to the forward tree $\BB_x$ defined as $H_{\kappa,x} := P_{x} H_\kappa P_{x}$. Here $P_x := P_{\BB_x} = \sum_{y\in\BB_x} |y\rangle\langle y|$ is the projection onto the forward tree and therefore $H_{\kappa,x}$ acts in $\ell^2(\BB_x)$. Then for any $x$ we see that also $H_{\kappa,x}$ has pure point spectrum at the root $x$ of the forward tree $\BB_x$. Secondly, assuming that we already know that these restrictions have pure point spectrum at their respective roots $x$ we prove pure point spectrum for $H_\kappa$ on the full Hilbert space $\ell^2(\BB)$.

Before we state the lemmas, let us introduce the spectral measures of the restricted Hamiltonian of an operator $H$ acting in $\ell^2(\BB)$. That is, suppose $\psi\in\ell^2(\BB)$. Then setting $\psi_x := P_x \psi\in\ell^2(\BB_x)$ we define the measure
\begin{equation}\label{eq-mu}
d\rho_{\psi_x} (E) := \wlim_{\epsilon\downarrow0} \frac1\pi \Im \langle \psi_x, (H_x-E-i\epsilon)^{-1} \psi_x\rangle\, dE,
\end{equation}
where $\mu=\wlim_{\epsilon\downarrow0} \mu_\epsilon$ is understood in the weak (or rather weak$^*$) sense, that is, $\int f(E)\, d\mu(E) = \lim_{\epsilon\downarrow0}\int f(E) \,d\mu_\epsilon(E)$ for all continuous, bounded functions $f$ on $\RR$.

We identify $\delta_x\in\ell^2(\BB)$ with $\delta_x\in \ell^2(\BB_x)$, and we say that an operator $H$ has pure point spectrum at $x$ if $\rho_{\delta_x}$ is a pure point measure.

Our first step is contained in the following lemma:
\begin{lemma} \label{pp at root} Under Assumptions \ref{assumptions1} there exists a $\kappa_0>0$ such that for $\kappa\in[0,\kappa_0]$ the operator $H_\kappa = -\Delta + Q_r + \kappa Q_{t}$ has almost surely pure point spectrum at the root $0$ of the full tree $\BB$.
\end{lemma}

\begin{proof} We begin by proving locally (in the energy) that $H_\kappa$ has pure point spectrum. To this end, let $E_0\in \mathcal S := [-2\sqrt{2}-K-1,2\sqrt{2}+K+1]$, where $\mathcal S$ is chosen to contain the spectrum of $H_\kappa$ for all $\kappa\le1$. Let $I\ni E_0,s,\epsilon_0,\kappa_0,C$, and $\ell$ be given as in Theorem \ref{mainthm}. Let $G_{\kappa,\omega}(0,x;E+i\epsilon) := \langle\delta_0,(H_{\kappa,\omega}-E-i\epsilon)^{-1} \delta_x\rangle$ denote the Green function of $H_{\kappa,\omega}$. Then, using Fatou's Lemma in the first step we arrive at the following estimate for $E\in I$,
\begin{eqnarray*}
\lefteqn{\EE\Big[\Big(\lim_{\epsilon\downarrow0} \sum_{x\in\BB} |G_{\kappa,\omega}(0,x;E+i\epsilon)|^2\Big)^{s/2}\Big]}
\\
&\le&\liminf_{\epsilon\downarrow0} \EE \Big[\sum_{x\in\BB} |G_{\kappa,\omega}(0,x;E+i\epsilon)|^2\Big)^{s/2}\Big]
\\
&=&\liminf_{\epsilon\downarrow0} \EE \Big[\Big(\sum_{n=0}^\infty\sum_{x\in S_n} |G_{\kappa,\omega}(0,x;E+i\epsilon)|^2\Big)^{s/2}\Big]
\\
&=&\liminf_{\epsilon\downarrow0} \EE \Big[\Big(\sum_{n=0}^\infty \|P_0(H_{\kappa,\omega}-E-i\epsilon)^{-1})P_n\|^2\Big)^{s/2}\Big]
\\
&\le&\liminf_{\epsilon\downarrow0} \sum_{n=0}^\infty \EE \big[\|P_0(H_{\kappa,\omega}-E-i\epsilon)^{-1})P_n\|^s\big]
\\
&\le& \liminf_{\epsilon\downarrow0} \sum_{n=0}^\infty C \ell^{-n}\;=\;C \Big(1-\frac1\ell\Big)^{-1},
\end{eqnarray*}
uniformly on the interval $I$. In the last inequality we have used our main technical result, Theorem~\ref{mainthm}, as well as the standard bound $(\sum_n |a_n|)^s \le \sum_n |a_n|^s$ for $s \in (0,1)$ (which will be used many more times below). This implies
\begin{equation*} 
\EE \Big[\int_I \Big(\lim_{\epsilon\downarrow0} \sum_{x\in\BB} |G_{\kappa,\omega}(0,x;E+i\epsilon)|^2\Big)^{s/2}\,dE\Big]
\;=\int_I\EE\Big[\Big(\lim_{\epsilon\downarrow0} \sum_{x\in\BB} |G_{\kappa,\omega}(0,x;E+i\epsilon)|^2\Big)^{s/2}\,dE\Big] \;\le \;C I
\end{equation*}
for some constant $C$. This shows that for $\PP-$almost all $\omega\in\Omega$ 
\begin{equation*}
\int_I \Big(\lim_{\epsilon\downarrow0} \sum_{x\in\BB} |G_{\kappa,\omega}(0,x;E+i\epsilon)|^2\Big)^{s/2}\, dE < \infty.
\end{equation*}
So for $\PP$-almost all $\omega\in\Omega$, there is a set of full Lebesgue measure in $I$ such that for $E$ in this set,
\beq \label{Simon-Wolff}
\lim_{\epsilon\downarrow0}\sum_{x\in\BB} |G_{\kappa,\omega}(0,x;E+i\epsilon)|^2 < \infty.
\ee
This is the Simon-Wolff criterion \cite[Theorem 8]{SW} which implies that $H_{\kappa,\omega}(\alpha) := H_{\kappa,\omega} + \alpha P_0$ has pure point spectrum in $I$ at the root for Lebesgue-almost all $\alpha\in\RR$. From this it is standard to show that $H_\kappa$ has almost surely pure point spectrum in $I$ at the root since the probability measure, $\nu_0$, at the root is purely absolutely continuous. 

The local result gives an open interval $I$ containing $E_0$ for every $E_0\in\mathcal S$ with corresponding coupling constant $\kappa_0(I)$. The intervals cover the set $\mathcal S$ and by compactness we can choose a finite sub-cover. We finish the proof by taking $\kappa_0$ to be the minimum of the corresponding values of $\kappa_0(I)$. 
\end{proof}

The next lemma allows to determine the spectral type of $H$ by finding the spectral type of $H_x$ at the root $x$ on arbitrary forward trees $\BB_x$.

\begin{lemma}\label{pp induction}
Let $V:\BB \to \RR$ be a potential such that the operator $H=-\Delta+V$ is self-adjoint on $\ell^2(\BB)$. Furthermore, let $H_x$ be the restriction of $H$ to the forward tree $\BB_x$ with root $x$.
\begin{enumerate}[{\rm (i)}]
 \item If for all $x \in \BB$ the spectral measures $\rho_{\delta_x}$ of $H_x$ are pure point then the spectrum of $H$ is pure point.
 \item More generally, if for all $x\in\BB$ the spectral measures $\rho_{\delta_x}$ of $H_x$ are of the same measure type (i.e. all pure point, purely singular, purely continuous, purely a.c., purely s.c., have no s.c. component), then the whole spectrum of $H$ is of this type.
\end{enumerate}
\end{lemma}
\begin{proof}
We will prove (i), the proof of (ii) is completely analogous.  By the spectral theorem we have an orthogonal sum decomposition $\ell^2(\BB_x)=\Hh^{\mathrm{c}}_{x} \oplus \Hh^{\mathrm{pp}}_{x}$ such that $H_x$ leaves these spaces invariant and such that $H_x$ restricted to $\Hh^{\mathrm{pp}}_x$ has pure point spectrum and $H_x$ restricted to $\Hh^{\mathrm{c}}_x$ has purely continuous (absolutely and singular) spectrum.

The assumption that $\rho_{\delta_x}$ is a pure point measure can be rewritten as $\delta_x \in \Hh_x^{\mathrm{pp}}$ for all $x\in\BB$.  Now, let $\psi\in\Hh^{\mathrm{c}}_0$. We will prove by induction over the level $d(x,0)$ that: 
%%%
\begin{equation}\label{*}
P_x\psi = \psi_x \in \Hh_x^{\mathrm{c}}.
\end{equation}
Once this is proved, it follows that $\psi(x)=\langle \delta_x, \psi\rangle=\langle \delta_x, \psi_x\rangle=0$ as $\delta_x\in\Hh^{\mathrm{pp}}_x$ which is orthogonal to $\Hh_{x}^{\mathrm{c}}$. Therefore,  $\psi = \sum_{x\in\BB} \psi(x) \delta_x=0$ for any $\psi\in\Hh^{\mathrm{c}}_0$. Hence, $\Hh^{\mathrm{c}}_0$ is trivial which implies that the spectrum of $H$ is pure point.
 
For $d(x,0)=0$, i.e. $x=0$, \eqref{*} is trivially true, $\psi_0=\psi\in\Hh_0^{\mathrm{c}}$ by definition.
For the induction step let $d(0,x) = n+1\geq 1$. Then $x$ has a parent $y$ defined by $d(0,y)=n,\,d(y,x)=1$. Besides $x$, $y$ has another child, $x'\neq x$ such that $\BB_y=\{y\}\cup \BB_x\cup\BB_{x'}$. Accordingly, with $\Gamma_y$ defined as 
\be{eq-Gamma}
\Gamma_y := |y\rangle \langle x| + |y\rangle\langle x'|\,+\, |x\rangle\langle y| + |x'\rangle \langle y|
\ee
we have the decomposition
\begin{equation*} %\label{eq-H-decomp}
H_y = V(y)\oplus H_x \oplus H_{x'}  \,-\, \Gamma_y.
\end{equation*}
By the induction hypothesis, $\psi_y\in\Hh^{\mathrm{c}}_y$. As $\delta_y \in \Hh^{\mathrm{pp}}_y$ this means that for any measurable function $f$ one has
\be{eq-y-orth}
\langle \delta_y, f(H_y) \psi_y \rangle = 0\,,
\ee
in particular $\psi(y)=0$. Using $H_y \delta_y  = -\delta_x - \delta_{x'} + V(y) \delta_y$ and \eqref{eq-y-orth} one also has
\be{eq-Gamma-orth}
0=-\langle \delta_y, H_y f(H_y) \psi_y \rangle = \langle \delta_x+\delta_{x'} ,f(H_y) \psi_y \rangle\,.
\ee
Combining \eqref{eq-Gamma}, \eqref{eq-y-orth} and \eqref{eq-Gamma-orth} we find
\be{eq-Gamma-0}
\Gamma_y f(H_y) \psi_y = 0\,.
\ee
A standard application of the resolvent identity gives
\begin{equation*}
(H_y-z)^{-1} = \left[V(y)\oplus H_x \oplus H_{x'}- z\right]^{-1}\left[ \one- \Gamma_y (H_y-z)^{-1}\right],
\end{equation*}
which using \eqref{eq-Gamma-0} and \eqref{eq-y-orth} leads to
\begin{align*}
\langle \psi_y, (H_y-z)^{-1} \psi_y \rangle &=
 \langle \psi_y, ( V(y)\oplus H_x \oplus H_{x'}- z)^{-1} \psi_y \rangle \notag \\
 &= \langle \psi_x, (H_x-z)^{-1} \psi_x \rangle\,+\, \langle \psi_{x'}, (H_{x'}-z)^{-1} \psi_{x'}\rangle\;.
\end{align*}
We have employed the direct sum structure $\psi_y = \psi(y) \oplus \psi_x \oplus \psi_{x'}$ with $\psi(y)=0$ in the last equation. Using \eqref{eq-mu} this implies for the corresponding (positive) measures that
\begin{equation*}
\rho_{\psi_y} = \rho_{\psi_x} + \rho_{\psi_{x'}}.
\end{equation*}
Therefore, by the induction hypothesis, $\rho_{\psi_x}$ must be a continuous measure, thus $\psi_x \in \Hh^{\mathrm{c}}_x$ and the induction step is done.
\end{proof}

Before concluding with the proof of Theorem~\ref{smallkappathm} let us give two remarks:

\begin{enumerate}[(i)]
\item It is obvious that the proof immediately extends to any rooted tree $\TT$ with root $0\in\TT$. There $H_x$ would be the restriction of $H$ to $\TT_x$ which is the 
forward tree with root $x$ (the branch from $x$ through $0$ is cut) and $\psi_x$ would denote the corresponding projection of $\psi$.
The only difference is that the last equation in the proof generalizes to
$\rho_{\psi_y} = \sum_{x\in N(y)} \rho_{\psi_x}$ where $N(y)$ denotes all forward neighbors (or children) of $y$, i.e.
$N(y)=\{x\in \TT: d(0,x)=d(0,y)+1\,\wedge\,d(x,y)=1 \}$.
\item Note that this lemma also holds if the 'measure type' (such as pure point) only refers to a specific (fixed) interval $I$. In this case one would choose $\psi$ to lie inside the spectral projection 
$P_I \ell^2(\TT)$ and be of complementary measure type.
As before one would inductively obtain the same for $\psi_x$ and any $x$, which would again imply $\psi=0$.
\end{enumerate}

\begin{proof}[Proof of Theorem~\ref{smallkappathm}] By our assumptions on the measures, in particular the existence of bounded densities of compact support for the marginal measures $\mu_0$ and $\mu_1$, Lemma \ref{pp at root} shows that for every $x$ in the tree there exist a $\kappa_0(x)$ and a set $\widetilde\Omega_x\subseteq\Omega$ with $\PP_{c(x)}[\widetilde\Omega_x] = 1$ such that for $\omega\in\widetilde\Omega_x$ the corresponding restricted Hamiltonian $H_{\kappa,x}(\omega)$ has pure point spectrum at the root $x$ of the forward tree $\BB_x$. Let $\Omega_x:=\pi_x^{-1}(\widetilde\Omega_x)$ be the potentials on the whole tree whose restrictions to $\BB_x$ lie in $\widetilde\Omega_x$. As noted above, $\PP[\Omega_x]=1$ as well. Since there are only three different measures at the roots, namely $\nu_0$ at $0$ and the marginals $\mu_0$ or $\mu_1$ at $x\not=0$ there are only three values of $\kappa_0$ and we take the minimum of them.

For a common set of events we choose $\Omega_\infty := \bigcap_{x\in\BB} \Omega_x$.  This is a set of full $\mathbb P$-measure for which we can apply Lemma \ref{pp induction}. This shows that $H_{\kappa,\omega}$ has pure point spectrum at $x$ for all $\omega\in\Omega_\infty$ meaning that the spectral measure of $H_{\kappa}$ in the states $\delta_x$, and hence in any state, is pure point.
\end{proof}

\section{Reduction of  Theorem~\ref{mainthm} to a dynamical system}\label{sectionreduction}

The idea is to rewrite the (fractional) moment $\mathbb E\left[\|P_0 (H_{\kappa}-E-i\epsilon)^{-1} P_n\|^s\right]$ in terms of a dynamical system, resulting in the expression \eqref{dynamical system bound}. To this end, we introduce for $z\in\CC^+$, $\kappa \ge 0$ and for any $x$ in the $n$th sphere $S_n$ the forward Green function at $x$ defined as $g_{x} := g_{x}(z) := \langle \delta_x, (H_{\kappa,x}-z)^{-1} \delta_x\rangle$. Note that $g_{0}(z) = G_{\kappa}(0,0;z)$ is the full diagonal Green function at $0$. Let $[0=x_0,x_1,\ldots,x_n=x]$ be the unique path of (connected) vertices from $0$ to $x$. Then, by an application of a resolvent identity it is well-known that these forward Green functions satisfy the recursion relation
\begin{equation}\label{recursion relation}
g_{x}(z) = -\frac1{g_{x'}(z) + g_{x''}(z) + z - Q(x)},
\end{equation}
where $x'$ and $x''$ are the two forward neighbors of $x$. Moreover, the full Green function $G_{\kappa}(0,x;z) = \langle \delta_0, (H_{\kappa}-z)^{-1} \delta_x \rangle$ can be expressed in the product form 
\begin{equation} \label{full Green}
G_{\kappa}(0,x;z) = \prod_{j=0}^n g_{x_j}(z).
\end{equation}

For fixed $z\in\CC^+$, $\omega \in \Omega$, $\kappa \ge 0$ and $n\ge 1$, the map $S_n\ni x\mapsto g_{x}(z)\in\CC^+$ takes on two possible values, and we denote these two values by $g_{n0}:=g_{n0}(z)$ and $g_{n1}:=g_{n1}(z)$. Therefore,
\begin{align} \|P_0 (H_{\kappa}-z)^{-1} P_n\|^2 &= \|P_0 (H_{\kappa}-z)^{-1} P_n (H_{\kappa}-\bar{z})^{-1}P_0\| \notag
=\sum_{x\in S_n} \big|\langle \delta_0,(H_{\kappa}-z)^{-1} \delta_x\rangle \big|^2 \notag
\\ &=\sum_{x\in S_n} |G_{\kappa}(0,x;z)|^2\notag
\\
&=\sum_{y\in S_{n-1}}  \big[|G_{\kappa}(0,y';z)|^2 +  |G_{\kappa}(0,y'';z)|^2  \big]\notag
\\
&=\sum_{y\in S_{n-1}} |G_{\kappa}(0,y;z)|^2  \big(|g_{n0}(z)|^2 +  |g_{n1}(z)|^2\big)\notag
\\
&= |g_0(z)|^2 \prod_{j=1}^n \big(|g_{j0}(z)|^2 + |g_{j1}(z)|^2\big), \label{prod repr}
\end{align}
by induction. The sequence of pairs $(g_{n0},g_{n1})_{n\ge1}$ is a sequence of identically distributed (two-dimensional) random variables. In order to obtain a dynamical system in \emph{one} variable we introduce for $n\ge1$ the random variables (for fixed $z\in\CC^+$)
\begin{equation} \label{g plusminus}
g_{n\pm} := g_{n\pm}(z) := \frac1{\sqrt{2}} (g_{n0}(z)\pm g_{n1}(z))
\end{equation}
and the maps $\phi^\pm_{z,q}$ on $\overline{\CC^+} := \CC^+ \cup \overline{\RR}$, with $\overline{\RR} := \RR\cup \{i\infty\}$ being the one-point compactification of $\RR$ (in particular, $\overline{\CC^+}$ is compact), defined as
\begin{equation} \label{phi} 
\overline{\CC^+}\ni w\mapsto\phi^\pm_{z,q}(w) := \frac12 \Big(\frac{-1}{w+ \frac{z-q_0}{\sqrt{2}}} \pm \frac{-1}{w+ \frac{z-q_1}{\sqrt{2}}}\Big),\quad q=(q_0,q_1).
\end{equation} 
We set $\phi^{\pm}_{z,q}(i\infty) = 0$. Then $|g_{n0}|^2 + |g_{n1}|^2 = |g_{n+}|^2 + |g_{n-}|^2$ and, due to                         (\ref{recursion relation}),  $g_{n\pm} = \phi^\pm_{z,q_n}(g_{(n+1)+})$ with $q_n = (q_{n0},q_{n1})$. For $1\le k\le n$ and $\mathbf{q} := \big((q_{10},q_{11});(q_{20},q_{21});\ldots;(q_{n0},q_{n1})\big)$ collecting all pairs of potentials from the first to the $n$th sphere, let us furthermore introduce the maps $\Phi^{\pm,k}_{z,\mathbf{q}}$ defined on $\overline{\CC^+}$ as
\begin{eqnarray*}
\Phi^{\pm,n}_{z,\mathbf{q}} &:=& \phi^{\pm}_{z,q_n},
\\
\Phi^{\pm,n-1}_{z,\mathbf{q}} &:=& \phi^{\pm}_{z,q_{n-1}}\circ\phi^+_{z,q_n},
\\
\vdots&&
\\
\Phi^{\pm,1}_{z,\mathbf{q}} &:=& \phi^\pm_{z,q_1}\circ\phi^+_{z,q_2}\circ\cdots\circ\phi^+_{z,q_n}.
\end{eqnarray*}
Then, with $w := g_{(n+1)+}$ and $1\le k\le n$,
\beq 
g_{k\pm} = \Phi^{\pm,k}_{z,\mathbf{q}}(w).
\ee
Using the representation \eqref{prod repr}, the expectation of $\|P_0(H_{\kappa}-z)^{-1} P_n\|^s$ can be taken iteratively starting with the expectation $\EE_{q_0}$ with respect to the potential $q_0$ at the root, then the expectation $\EE_{q_1}$ with respect to the potentials $q_{10}$ and $q_{11}$ at the first sphere, and so forth until we finally take the expectation $\EE_{q_n}$ with respect to the potentials $q_{n0}$ and $q_{n1}$ at the $n$th sphere. In the simplest case, every time taking such an expectation $\EE_{q_k}$, we lower the previous expectation value by a factor $\delta<1$. In order to claim that the whole product is exponentially decaying in $n$ we need the initial step that the starting point is finite. The latter follows by standard arguments which will be used repeatedly. Here, it is important that $s\in(0,1)$. For any $a\in\RR$, let $a^s := \mathrm{sgn}{(a)}|a|^{s}$. We do not have any a-priori information on $g_0$ or the forward Green functions but we may use the recursion relation (\ref{recursion relation}) and take as an upper 
bound the supremum over all of $\overline{\CC^+}$. Then, using that $\Im(g_{0'}(z)+g_{0''}(z)) \ge 0$,
\begin{eqnarray} \label{initial bound}
\EE_{q_0}\big[|g_0|^s\big] &\le& \sup_{w\in\overline{\CC^+}} \int_{-K}^K \frac{\nu_0(q)\, dq}{|w+z-q|^s} \notag
\\
&\le& \|\nu_0\|_\infty \sup_{w\in\overline{\CC^+}} \int_{-K}^K \frac{dq}{\big[(\Re(w)+\Re(z)-q)^2+(\Im(w)+\Im(z))^2\big]^{s/2}}\notag
\\
&\le& \|\nu_0\|_\infty \sup_{w\in\RR} \int_{-K}^K \frac{dq}{\big[(w-q)^2\big]^{s/2}}\notag
\\
&=& \|\nu\|_{\infty} \int_{-K}^K \frac{dq}{|q|^s} \notag
\\
& = & \frac{2\|\nu\|_{\infty}}{1-s} K^{1-s}. 
\end{eqnarray}

Before we finally introduce the formulation in terms of a dynamical system let us put down the bound on the fractional moment that we have obtained so far. We will use again $w:=g_{(n+1)+}$. Then,
\begin{eqnarray}
\lefteqn{\EE\big[\|P_0 (H_{\kappa}-z)^{-1} P_n\|^s\big]}\notag
\\
&=&\EE\Big[|g_0|^s \prod_{j=1}^n \big(|g_{j+}|^2 + |g_{j-}|^2\big)^{s/2}\Big]\notag
\\
&\le&\EE\Big[\Big(\big|\Phi^{+,n}_{z,\mathbf{q}}(w)\big|^s + \big|\Phi^{-,n}_{z,\mathbf{q}}(w)\big|^s\Big) \cdots \Big(\big|\Phi^{+,1}_{z,\mathbf{q}}(w)\big|^s + \big|\Phi^{-,1}_{z,\mathbf{q}}(w)\big|^s\Big)\,|g_0|^s \Big]\notag
\\
&\le&C\sup_{w\in\overline{\CC^+}}\EE_{q_n}\Big[\big|\Phi^{+,n}_{z,\mathbf{q}}(w)\big|^s + \big|\Phi^{-,n}_{z,\mathbf{q}}(w)\big|^s\Big]\cdots\EE_{q_1}\Big[\big|\Phi^{+,1}_{z,\mathbf{q}}(w)\big|^s + \big|\Phi^{-,1}_{z,\mathbf{q}}(w)\big|^s\Big],
\end{eqnarray}
where $\mathbf{q} = (q_1,q_2) = (r+\kappa p_1, r+\kappa p_2)$ and we have used in the last step that the first expectation with respect to the potential at the root is bounded according to \eqref{initial bound}. As mentioned above, the iterated expectation $\EE_{q_n}\cdots \EE_{q_1}$ can be interpreted as a dynamical system run backwards, starting with $n$ down to $1$ and using that $\Phi^{\pm,k}_{z,\textbf{q}} = \phi^{\pm}_{z,q_k}\circ \Phi^{+,k+1}_{z,\textbf{q}}$. To this end, let $C(\overline{\CC^+}), C^+(\overline{\CC^+})$ and $C(\overline{\RR})$ be the Banach space of complex-valued, continuous functions on $\overline{\CC^+}$ (i.e.\ the continuous functions on $\CC^+ \cup \RR$ which have a finite limit at $\infty$), the cone of non-negative functions in $C(\overline{\CC^+})$, and the Banach space of real-valued, continuous functions on $\overline{\RR}$, respectively, all equipped with the sup-norm, referred to in the following by
$\|\cdot\|_{\infty}$. On these spaces, we introduce the operator $T_{\kappa,z,s}$. We will show in Lemma~\ref{prelim properties of T} below that this operator is well-defined.

\begin{defini}[Dynamical system]
For $z\in\CC^+\cup\RR$, $\kappa \ge 0$ and $s\in(0,1)$, the linear operator $T_{\kappa,z,s}$ maps a function $f\in C(\overline{\CC^+})$ to a function $T_{\kappa,z,s} f$ defined also on $\overline{\CC^+}$ in the following way {\rm (}see \eqref{phi} for the definition of the maps $\phi^{\pm}_{z,q}${\rm )}:
\begin{align} \label{def of T}
\lefteqn{\overline{\CC^+}\ni w\mapsto (T_{\kappa,z,s}f)(w) \;:=\; \EE_q\Big[f\big(\phi^+_{z,q}(w)\big) \, \big(|\phi^+_{z,q}(w)|^s + |\phi^-_{z,q}(w)|^s\big)\Big]}
\\
&=\iint\Big[f\big(\phi^+_{z,q}(w)\big) \, \big(|\phi^+_{z,q}(w)|^s + |\phi^-_{z,q}(w)|^s\big)\Big] \, d\mu(q)\notag
\\
&=\iiint\Big[f\big(\phi^+_{z,(r+\kappa p_0,r+\kappa p_1)}(w)\big) \, \big(|\phi^+_{z,(r+\kappa p_0,r+\kappa p_1)}(w)|^s + |\phi^-_{z,(r+\kappa p_0,r+\kappa p_1)}(w)|^s\big)\Big] \, \nu(r) dr\, d\sigma(p_0,p_1)\notag.
\end{align}
Since we have set $\phi^{\pm}_{z,q}(i\infty) = 0$ we define $(T_{\kappa,z,s}f)(i\infty) := 0$ and $T_{\kappa,i\infty,s} := \nul$. For $E\in\RR$, we use the same operator symbol, $T_{\kappa,E,s}$, to denote the operator that maps $C(\overline{\RR})$ into itself.
\end{defini}

Then, with $1$ denoting the constant function equal to 1 (and assuming that $T_{\kappa,z,s}$ and powers of $T_{\kappa,z,s}$ are well-defined, see Lemma~\ref{prelim properties of T}) we have
\beq 
\EE_{q_n}\Big[\big|\Phi^{+,n}_{z,\mathbf{q}}(w)\big|^s + \big|\Phi^{-,n}_{z,\mathbf{q}}(w)\big|^s\Big]\cdots\EE_{q_1}\Big[\big|\Phi^{+,1}_{z,\mathbf{q}}(w)\big|^s + \big|\Phi^{-,1}_{z,\mathbf{q}}(w)\big|^s\Big] = 
(T_{\kappa,z,s}^n 1)(w).
\ee
We summarize this reduction in the following proposition.  
\begin{prop}\label{prop:reduction to FM}
Let $z\in\CC^+,s\in(0,1)$ and let $P_n$ denote the projection onto the $n$th sphere. Furthermore, assume Assumptions \ref{assumptions1} on the measures $\nu_0,\nu$ and $\sigma$. Then, the fractional moment of the Green function of $H_{\kappa} = -\Delta + Q_r + \kappa Q_{t}$ is bounded in terms of the dynamical system defined through the operator $T_{\kappa,z,s}$ as in \eqref{def of T} by
\beq\label{dynamical system bound}
\EE\big[\|P_0 (H_{\kappa}-z)^{-1} P_n\|^s\big] \le C \sup_{w\in\overline{\CC^+}} (T_{\kappa,z,s}^n 1)(w) = C \|T_{\kappa,z,s}^n 1\|_\infty.
\ee
The constant satisfies $C\le 2K^{1-s}\|\nu_0\|_\infty/(1-s)$.
\end{prop}

Thus the proof of Theorem~\ref{mainthm} is reduced to showing exponential decay of $\|T_{\kappa,z,s}^n 1\|_\infty$ for sufficiently small positive $s$. Accordingly, the remainder of this work is a careful study of the dynamical system defined by $T_{\kappa,z,s}$ on $C(\overline{\CC^+})$. We start by collecting some important technical properties of the operator $T_{\kappa,z,s}$.

\begin{lemma} \label{prelim properties of T}
We assume that the probability measures $\nu$ and $\sigma$ satisfy Assumptions \ref{assumptions1}. Let $s\in(0,1/2)$. Then the following holds.
\begin{enumerate}[{\rm (i)}]
\item \label{1} For $z\in\CC^+\cup\RR$ and $\kappa \ge 0$, $T_{\kappa,z,s}:C(\overline{\CC^+})\to C(\overline{\CC^+})$ is a bounded linear operator with uniformly (in $z$ and $\kappa$) bounded operator norm, 
\beq\label{uniform operator norm bound}
\|T_{\kappa,z,s}\|_\infty\le \frac{4 \cdot 2^{1-s/2} \|\nu\|_{\infty} K^{1-s}}{1-s}.
\ee
\item \label{2} For every $f\in C(\overline{\CC^+})$, the map $z \mapsto T_{\kappa,z,s}f$ is continuous from $\CC^+ \cap \RR$ into $C(\overline{\CC^+})$ equipped with the sup-norm.
\item \label{2b}
Let $\{f_{\alpha}: \alpha \in \mathcal{J}\}$ be equi-continuous and bounded in $C(\overline{\CC^+})$ and $m\ge 1$ integer. Then
\beq \label{unicontbound}
\sup_{\kappa \ge 0} \sup_{E\in \RR} \| \left( T_{\kappa,E+i\epsilon,s}^m - T_{\kappa,E,s}^m \right) f_{\alpha}\|_{\infty} \rightarrow 0
\ee
as $\epsilon \downarrow 0$ uniformly in $\alpha \in \mathcal{J}$.
\item \label{3} For every $z\in \CC^+\cup\RR$, the map $T_{\kappa,z,s}$ preserves the cone $C^+(\overline{\CC^+})$ of non-negative functions. Hence, for two real-valued functions $f$ and $g$ with $f\le g$ (point-wise) we have $T_{\kappa,z,s} f \le T_{\kappa, z,s} g$ (point-wise). 
\item \label{4} For any positive integer $m$ and any $z\in \CC^+\cup\RR$ we have
\begin{equation} \label{sup over R}
\sup_{w\in\overline{\CC^+}} (T_{\kappa,z,s}^m 1)(w) = \sup_{w\in\overline{\RR}} (T_{\kappa,z,s}^m 1)(w) . %= \sup_{w\in\RR} (T_{z}^m f)(w).
\end{equation}
\end{enumerate}
\end{lemma}

We defer the proof of these properties to Section~\ref{appendix} at the end of this paper, but include a few comments here. 

As will be seen in the proofs, the assumption $s<1/2$ is used in the proofs of parts (\ref{2}) and (\ref{2b}), while the other parts hold more generally for $s<1$.

Part (\ref{2b}) is a strengthening of the strong continuity property of $T_{\kappa,z,s}$ proved in part (\ref{2}) (at least at real energy $E$), showing that this continuity is uniform in the parameters $E$ and $\kappa$ and also holds for powers of the $T$-operators. We will use part (\ref{2b}) in the proof of Theorem~\ref{mainthm} with the trivial one-element family $\{f_{\alpha}\} = \{1\}$. That we state part (\ref{2b}) for general equi-continuous families $\{f_{\alpha}\}$ here is prompted by the technique used in its inductive proof in Section~\ref{appendix}.

In Proposition \ref{prop:reduction to FM} we have introduced the upper bound on the fractional moment in terms of a dynamical system on $C(\overline{\CC^+})$. Lemma \ref{prelim properties of T}\eqref{4} tells us that we may reduce to a dynamical system on $C(\overline{\RR})$. This was important for numerical computations and will also be used in the proof of Theorem~\ref{mainthm} in Section~\ref{sectionproof2} below. 

%, respectively on $\RR$. For technical reasons (such as continuous functions on $\overline{\RR}$ are uniformly continuous), we prefer in the following to consider the compact space $\overline{\RR}$.

For real $E$ and any positive integer $m$, we will also need the strong continuity of $(T_{\kappa,E,s})^m$ on $C(\overline{\RR})$ as a function of $\kappa$, the coupling constant at the transversally periodic potential $Q_t$. This is the content of the next result.  In this context, $\|\cdot\|_{\infty}$ refers to the sup-norm on $C(\overline{\RR})$. 

\begin{lemma} \label{strong continuity of T to m in kappa}
Under Assumptions \ref{assumptions1} on the measures $\nu_0$, $\nu$ and $\sigma$, $s\in(0,1/2)$ and for any $m\in\NN$ we have strong continuity of $\kappa\mapsto(T_{\kappa,E,s})^m$ at $0$, uniformly in $E$. That is, for $f\in C(\overline{\RR})$, 
\beq
\lim_{\kappa\downarrow0} \sup_{E\in\RR} \big\|(T_{\kappa,E,s})^m f - (T_{0,E,s})^mf \big\|_\infty = 0.
\ee
\end{lemma}

Again, we defer the proof to Section~\ref{appendix}.

\section{Estimate for zero coupling}\label{sectionkappazero}

In this section we set $\kappa=0$ and consider the operator $T_{E,s}:C(\overline\RR)\rightarrow C(\overline\RR)$ defined by
\begin{align*}
(T_{E,s}f)(w) &= \EE_r \left[f(\phi^+_{E,r}(w))\,|\phi^+_{E,r}(w)|^s\right]\\
&= \int f\left(\frac{-1}{w+(E-r)/\sqrt{2}}\right)\left|\frac{-1}{w+(E-r)/\sqrt{2}}\right|^{s}\nu(r) dr\\
&= \int f\left(\frac{-1}{w+\widetilde{E}-q}\right)\left|\frac{-1}{w+\widetilde{E}-q}\right|^{s}\widetilde{\nu}(q) dq,
\end{align*}
where $\widetilde{E}=E/\sqrt{2}$ and $\widetilde{\nu}(q) = \sqrt{2}\,\nu(\sqrt{2} q)$. For the remainder of this section, we drop the tildes and let $\EE$ denote the expectation with respect to the rescaled measure.
Our goal in this section is to prove the following lemma up to a key estimate, \eqref{eq-est-0}, whose proof is the topic of the next section.

\begin{lemma}\label{onedimlemma}
For every $E_0\in\RR$ there exist an open interval $I$ containing $E_0$, $s\in (0,1/2)$ and $m\in\NN$ such that
\begin{equation} \label{strictcontract}
\sup_{E\in I} \|T_{E,s}^m 1\|_\infty < 1.
\end{equation}
\end{lemma}

\begin{remark}\label{rm:1D-proof}
Before proceeding, let us briefly explain how the methods of this paper yield a new proof of localization for the one-dimensional Anderson model, as was first proved under the assumptions used here in \cite{KS} by a different method. Thus let $H=-\Delta + q$ in $\ell^2(\ZZ)$ with i.i.d.\ random potential $q$ whose distribution has a bounded and compactly supported density. For $m<n$ the resolvent of $H$ satisfies an analogue of (\ref{full Green}),
\[
\langle \delta_m, (H-z)^{-1} \delta_n \rangle = \langle \delta_m, (H-z)^{-1} \delta_m \rangle \prod_{j=m+1}^n \langle \delta_j, (H_j-z)^{-1} \delta_j \rangle,
\]
where $H_j$ is the restriction of $H$ to $\ZZ \cap [j,\infty)$. As in Proposition~\ref{prop:reduction to FM} this leads to the bound $\EE(|\langle \delta_m, (H-z)^{-1} \delta_n \rangle|^s) \le C \|T_{z,s}^{n-m} 1\|_{\infty}$, where, in this context, $T_{z,s}$ is the simplified dynamical system given by (\ref{1Ddynsyst}). By an analogue of Lemma~\ref{onedimlemma} this leads to the fractional moments bound
\[
\sup_{E\in I, \epsilon>0} \EE(| \langle \delta_m, (H-E-i\epsilon)^{-1} \delta_n \rangle|^s) \le C \ell^{-(n-m)}
\]
for some $\ell>1$, using also the arguments in Section~\ref{sectionproof2} below.  For the one-dimensional case this is known to imply spectral as well as dynamical localization, e.g.\ \cite{Sto}.
\end{remark}

To motivate the following proof of Lemma~\ref{onedimlemma}, let us start by discussing some previous work. Consider the sequence of positive functions $(T_{E,s}^m 1)(w)$ for $m=0,1,2,\ldots$. We want to show that the sup norm is eventually below $1$. For large disorder, corresponding to the operator $-\Delta + \lambda Q_r$ with $|\lambda|$ large, the integration over $Q$ achieves this after one step. This is essentially the Aizenman-Molchanov proof of high disorder localization in one dimension. For arbitrary disorder, the sup norms of the sequence may initially grow, but numerical experiments \cite{L} suggested that they always eventually decay to zero. In order to prove this, we first observe that the monotonicity property, Lemma \ref{prelim properties of T}(iii), implies that if iterates of a positive bounding function decrease to zero, then iterates of any positive function lying below the bounding function decrease to zero as well. We look for such a bounding function in the class of functions 
\begin{equation}\label{boundingfcn}
\varphi_\zeta(w)^s := |w-\zeta|^{-s}
\end{equation}
indexed by $\zeta\in \CC^+$. Then $T_{E,s}$ decreases $\varphi_\zeta(w)^s$ after one step if $F_\zeta(w,s,E)<1$ where
\begin{align}
F_\zeta(w,s,E) :&= \varphi_\zeta(w)^{-s}\left(T_{E,s}\varphi_\zeta(\cdot)^{s} \right)(w)\notag\\ 
&= \EE\left[\left|\frac{w-\zeta}{\zeta(w+1/\zeta+E-q)}\right|^s\right].\label{decg}
%&= \int\left|\frac{w-\zeta}{\zeta(w+1/\zeta+E-q)}\right|^s \nu(q) dq
\end{align}
In the limit of zero disorder $\nu$ converges to a delta function at $q=0$. In this case we can choose $\zeta_0$ to be one of the fixed points of $\zeta\mapsto -1/\zeta -E$.
For $|E|\leq 2$ these lie on the unit circle and we find $F_{\zeta_0}(w,s,E)=1$,
for $|E|>2$ we can choose a fixed point on the real line with absolute value $>1$ so that $F_{\zeta_0}(w,s,E)<1$. For the critical energies $|E_0|\leq 2$ we can construct a proof for low disorder by computing a perturbation series in a disorder parameter. 
This was done in \cite{L}. In this paper, the key estimate \eqref{eq-est-0}, which will be shown in the next section, allows us to bound $F_\zeta(w,s,E)$ for small $s$ and $E$ near $E_0$ for any disorder.

We now begin the proof of Lemma \ref{onedimlemma}. We start with some estimates on the bounding function.

\begin{lemma}\label{gbddlemma}
Let $\varphi_\zeta(w)^s$ be the bounding function defined by \rf{boundingfcn}. Then
\begin{enumerate}[{\rm (i)}]
\item For every bounded interval $I\subset \RR$, $\zeta\in\CC^+$ and $s\in(0,1)$, there exists an $A<\infty$ such that $(T_{E,s}1)(w) \le A \varphi_\zeta(w)^{s}$ for all $E\in I$ and all $w\in\RR$.
\item $\varphi_\zeta(w)^s \le \Im(\zeta)^{-s}$ for all $w\in\RR$.
\end{enumerate}
\end{lemma}
\begin{proof}
We have
\begin{align*}
\varphi_\zeta(w)^{-s}(T_{E,s}1)(w) &= \EE\left[\left|\frac{w-\zeta}{w+E-q}\right|^s\right]\\
&\le 1 + \EE\left[\left|\frac{q-E-\zeta}{w+E-q}\right|^s\right]\\ 
&\le A(\zeta,K,I,s).
\end{align*}
uniformly in $E\in I$ and $w\in \RR$ (which uses boundedness of $I$ and supp$\,\nu$). This proves (i). Part (ii) is immediate.
\end{proof}

\begin{lemma}\label{decglemma}
For every $E_0\in\RR$ there exist a $\zeta_0\in\CC^+$ with $|\zeta_0|>1$, an open interval $I$ containing $E_0$ and $\delta<1$  such that
\begin{equation}\label{decg2}
\sup_{E\in I}\|\varphi_{\zeta_0}^{-s}\left(T_{E,s}\varphi_{\zeta_0}^{s} \right)\|_\infty \le \delta.\
\end{equation}
for small non-zero $s$.
\end{lemma}
\begin{proof}
The inequality \rf{decg2} can be written as
\begin{equation}\label{decg3}
\sup_{E\in I}\sup_{w\in\RR} F_{\zeta_0}(w,s,E) \le \delta,
\end{equation} 
where $F_\zeta$ is defined by \rf{decg}. The choice $\zeta \in \CC^+$ avoids singularities. Thus it is easy to see that $F_{\zeta}(w,s,E)$ is jointly continuous in $(w,s,E) \in \RR \times [0,1) \times \RR$ and
\begin{equation*} %\label{Fis1}
F_\zeta(w,0,E)=1
\end{equation*}
for all $\zeta\in \CC^+$, $w\in\RR$ and $E\in\RR$.
Furthermore, $F_\zeta(w,s,E)$ is differentiable in $s$,
$$
\left(\frac{\partial F_\zeta}{\partial s}\right)(w,0,E) = \EE \left[ \left| \frac{w-\zeta}{\zeta(w+1/\zeta +E-q)} \right|^s \log \left| \frac{w-\zeta}{\zeta(w+1/\zeta +E-q)} \right| \right]
$$
is jointly continuous on the same domain and
\begin{eqnarray}
\left(\frac{\partial F_\zeta}{\partial s}\right)(w,0,E) &= & \EE \left[ \log \left| \frac{w-\zeta}{\zeta(w+1/\zeta +E-q)} \right| \right] \label{logmoment}
\\
& = & -\frac{1}{2}\EE\left[\log \frac{|\zeta|^2 |w+1/\zeta+E-q|^2}{|w-\zeta|^2}\right]. \notag
\end{eqnarray}
The key estimate \eqref{eq-est-0} applied to $Q=q-E_0$ implies that there exist a $\zeta_0\in\CC^+$ with $|\zeta_0|>1$ and an $\epsilon_1 >0$ such that for all $w\in\RR$,
\begin{equation}\label{Fsneg}
\left(\frac{\partial F_{\zeta_0}}{\partial s}\right)(w,0,E_0) < -\epsilon_1.
\end{equation}
For this choice of $\zeta_0$ we now control $F_{\zeta_0}$ for large $w$ using equation \eqref{decg}. In order to do this, notice that
\begin{equation}
\left|\frac{w-\zeta_0}{\zeta_0(w+1/\zeta_0 +E-q)}\right|
\le \frac{1}{|\zeta_0|}\left(1 + \frac{|1/\zeta_0 +E-q+\zeta_0|}{|w+1/\zeta_0 +E-q|}\right).
\end{equation}
Therefore, since the support of $\nu$ is bounded and because $|\zeta_0|>1$, there exist a $\delta_1<1$ and a constant $W$ such that 
\begin{equation*}
\sup_{|E-E_0|\le 1} \sup_{|w|\ge W}\left|\frac{w-\zeta_0}{\zeta_0(w+1/\zeta_0 +E-q)}\right| 
\le \delta_1
\end{equation*}
for all $q \in \,\mbox{supp}\,\nu$, so that 
\begin{equation}\label{largew}
\sup_{|E-E_0|\le 1} \sup_{|w|\ge W} F_{\zeta_0}(w,s,E) \le \delta_1^s.
\end{equation}
We now control $F_{\zeta_0}$ for small $w$. The set $\left\{(w,s,E) : \frac{\partial F_{\zeta_0}}{\partial s}(w,s,E) < -\epsilon_1\right\}$ is open and contains $[-W,W]\times\{0\}\times\{E_0\}$. It therefore contains a set $[-W,W]\times\ [0,s_0)\times I$, for some open interval $I$ containing $E_0$. For $|w|\le W$, $s<s_0$ and $E\in I$ we then obtain
\begin{equation}\label{smallw}
F_{\zeta_0}(w,s,E) = 1 + \int_0^s \left(\frac{\partial F_{\zeta_0}}{\partial s}\right)(w,s',E)ds' \le 1-\epsilon_1 s.
\end{equation}
Shrinking $I$ if needed so that $|E-E_0|\le 1$ for $E\in I$, we can combine \rf{smallw} with \rf{largew} to obtain \rf{decg3} for small non-zero $s$.
\end{proof}

\begin{rem} 
The expression (\ref{logmoment}) can be viewed as a {\it logarithmic moment} of the random variable $q$, which appears in our method as a limiting case of fractional moments for $s\to 0$. This is not surprising, as it has been observed previously that applying the fractional moment method in 1D requires choosing $s$ close to $0$, e.g.\ \cite{Hamzaetal}.
\end{rem}

We can now prove the main result of this section.

\begin{proof}[Proof of Lemma \ref{onedimlemma}] Given $E_0\in\RR$ we choose $\zeta_0$, $I$, $s$ and $\delta<1$ according to Lemma~\ref{decglemma} and let $A$ be the corresponding bound from Lemma \ref{gbddlemma}. Then, using the monotonicity property, Lemma \ref{prelim properties of T}(iii), of $T_{E,s}$ we find that for $E\in I$,
$$
\|T_{E,s}^m 1\|_\infty \le A\|T_{E,s}^{m-1} \varphi_{\zeta_0}^s\|_\infty
\le A\delta^{m-1}\|\varphi_{\zeta_0}^s\|_\infty \le A\delta^{m-1}(\Im \zeta_0)^{-s}.
$$
Choosing $m$ sufficiently large completes the proof.
\end{proof}

\section{Proof of the key estimate}\label{sectionkey}

The only part missing in the proof of Lemma~\ref{onedimlemma} (and Lemma~\ref{decglemma}) is a justification for the inequality \eqref{Fsneg}. 
In fact, this inequality follows from the following theorem which is the main goal of this section.

\begin{theorem}\label{keyestthm}
 Let $Q$ be a bounded real valued random variable
 whose distribution is supported on at least 3 points (i.e. the distribution of
 $Q$ is not a single or the sum of two delta measures).
 Then there exists a $\zeta\in \CC^+$ with $|\zeta|>1$ such that 
\begin{equation}\label{eq-est-0}
\inf_{w\in\overline\RR} \EE\left[\log \frac{|\zeta|^2|w+1/\zeta-Q|^2}{|w-\zeta|^2}\right] > 0\;.
\end{equation}
\end{theorem}

In order to prove this, we first define for $b>0$
\begin{equation}
\label{eq-def-alpha}
\alpha:=\alpha(a+bi,Q) := \frac{1}{b}\left(a+\frac{a}{a^2+b^2} - Q \right) + \frac{i}{a^2+b^2}\,.
\end{equation}
Then one finds
\begin{equation}
\label{eq-est-1}
\inf_{w\in\overline\RR} \;\EE\left[ \log \frac{|\zeta|^2|w+1/z-Q|^2}{|w-\zeta|^2}\right] = 
\inf_{u\in\overline\RR} \;\EE\left[ \log \frac{(\Im\alpha(\zeta,Q))^{-1}|u+\alpha(\zeta,Q)|^2}{|u+i|^2}\right],
\end{equation}
where $u$ and $w$ are related by $w=bu+a$ with $\zeta=a+ib$. 

Define $f: \CC^+ \rightarrow \RR$ by
\begin{equation}
\label{eq-def-f}
f(z) := -\log\left(\Im\Big(\frac{-1}{\;z}\Big)\right) = \log\left(\frac{|z|^2}{\Im(z)} \right)
\end{equation}
then by \eqref{eq-est-1} the inequality \eqref{eq-est-0} is equivalent to
\begin{equation}\label{eq-est-f}
\inf_{u\in\overline\RR} \EE\left[ f(u+\alpha(\zeta,Q))\right] - f(u+i)\,>\,0\;.
\end{equation}
The main argument will be based on convexity in the hyperbolic upper half plane. 
However, there is a much simpler argument for the special case $Q=q-E$ where $E\in\RR$ and $q$ is distributed symmetrically about zero, with possibly unbounded $q$. We present a sketch of this argument first.

\begin{lemma}\label{symmzeta} Let $Q=q-E$ where $E\in\RR$, $q$ is distributed symmetrically about zero, $\EE[\log(1+q^2)]<\infty$ and $\EE(q^2)>0$. Then, there exists a $\zeta\in\CC^+$ with $|\zeta|>1$ such that
\begin{equation}\label{realzeta}
\Re(\zeta + 1/\zeta + E) = 0
\end{equation}
and
\begin{equation}
\EE\big[\log |\alpha(\zeta,q-E)|^2 \big] = 0.
\end{equation}
\end{lemma}
\begin{proof}
Let $d\in (0,\infty)$ and take $\zeta(d)$ to be the root of $\zeta + 1/\zeta + E=id$ of largest imaginary part. We obtain a continuous curve $(0,\infty)\ni d\mapsto\zeta(d)\in\CC^+$ of solutions to \eqref{realzeta} with $\lim_{d\uparrow \infty}\Im\zeta(d)=\infty$. 

We have
\begin{equation}\label{symmalpha}
\alpha(\zeta(d),q-E) = \frac{-q}{\Im\zeta(d)} + \frac{i}{|\zeta(d)|^2}
\end{equation}
so that
\begin{equation}\label{expalpha}
\EE\big[\log |\alpha(\zeta(d),q-E)|^2 \big] = \EE\left[\log\left(\frac{q^2}{\big(\Im\zeta(d)\big)^2} + \frac{1}{|\zeta(d)|^4}\right)\right],
\end{equation}
which goes to $-\infty$ as $d\to \infty$. 

If $|E|\ge 2$, then $\lim_{d\to 0} \Im\,\zeta(d)=0$, so that (\ref{expalpha}) goes to $\infty$ as $d\to 0$. If, on the other hand, $|E|<2$, then $\lim_{d\to 0} \zeta(d) = (-E+i\sqrt{4-E^2})/2$ lies on the unit circle. Thus, as $d\to 0$, (\ref{expalpha}) approaches
$$
\EE \left[ \log \left( \frac{q^2}{1-E^2/4} + 1 \right) \right] \ge \EE \left[ \log(q^2+1) \right] >0,
$$
which uses that $q^2$ is positive with positive probability as $\EE(q^2)>0$. 

By continuity, in each case there is a value of $d$ for which the expectation is zero. Again, since $q^2$ is not supported at $0$, equation \eqref{expalpha} also implies that the corresponding value of $\zeta(d)$ has $|\zeta(d)| > 1$.
\end{proof}
Given this lemma we can now prove the special case of the theorem.
\begin{prop} Under the hypotheses of Lemma \ref{symmzeta}, let $\zeta$ be the value given by the lemma. We further assume that the distribution of $q^2$ is supported on at least two points (i.e. the variance of $q^2$ is positive).
Then the key inequality \eqref{eq-est-0} holds.
\end{prop}

\begin{proof} We will show that \eqref{eq-est-f} holds. Let $\zeta$ be the value given by Lemma \ref{symmzeta} so that $\alpha = \alpha(\zeta,q-E)$ is given by \eqref{symmalpha}. From (\ref{realzeta}) we find $\alpha(\zeta,-q-E)=-\overline\alpha(\zeta,q-E)$. This and symmetry of the distribution of $q$ give
\begin{align*}
\EE\left[\log\frac{|u+\alpha|^2}{\Im\alpha}\right]
=\EE\left[\log\frac{|u-\overline\alpha|^2}{\Im(-\overline\alpha)}\right]
=\EE\left[\log\frac{|u-\overline\alpha|^2}{\Im\alpha}\right].
\end{align*}
Hence
\begin{align*}
\EE[f(u + \alpha)] &= \EE\left[\log\frac{|u+\alpha|^2}{\Im\alpha}\right]\\
&=\frac{1}{2}\EE\left[\log\frac{|u+\alpha|^2|u-\overline\alpha|^2}{(\Im\alpha)^2}\right]\\
&=\frac{1}{2}\EE\left[\log\frac{|u^2-|\alpha|^2+2iu\Im\alpha|^2}{(\Im\alpha)^2}\right]\\
&=\frac{1}{2}\EE\left[\log\frac{(u^2-|\alpha|^2)^2+4u^2(\Im\alpha)^2}{(\Im\alpha)^2}\right]\\
&=\frac{1}{2}\EE\left[\log\frac{(u^2-|\alpha|^2)^2}{(\Im\alpha)^2}+4u^2\right]\\
&\ge\frac{1}{2}\EE\left[\log\frac{(u^2-|\alpha|^2)^2+4|\alpha|^2u^2}{|\alpha|^2}\right]\\
&=\frac{1}{2}\EE\left[\log\frac{(u^2+|\alpha|^2)^2}{|\alpha|^2}\right]\\
&=-\frac{1}{2}\EE\left[\log|\alpha|^2\right] + \EE\left[\log(u^2+|\alpha|^2)\right]\\
&=\EE\left[\log(u^2+|\alpha|^2)\right].
\end{align*}
In the last step we used that $\EE\left[\log|\alpha|^2\right]=0$. The function $s\mapsto \log(u^2+e^s)$ is strictly convex for $u\ne 0$. So for $u\ne 0$, since $|\alpha|^2 = q^2/(\Im\zeta)^2 + 1/|\zeta|^4$ is supported on at least two points we have a strict inequality in Jensen's inequality:
\begin{align*}
\EE\left[\log(u^2+|\alpha|^2)\right] &= \EE\left[\log(u^2+e^{\log |\alpha|^2})\right]\\
&>\log\Big(u^2+e^{\EE\left[\log |\alpha|^2\right]}\Big)\\
&=\log(u^2+1)\\
&= f(u+i).
\end{align*}
Thus we have proved the strict inequality $\EE[f(u+\alpha)]>f(u+i)$ for $u^2\ne 0$. In order to obtain the uniform statement \eqref{eq-est-f} it remains to check the strict inequality at $u=0$ and in the limit $|u|\rightarrow\infty$. At both these endpoints the strict inequality follows from \eqref{eq-def-alpha}, \eqref{eq-est-1} and the fact that $1/(\Im\alpha) = |\zeta|^2>1$. For $u=0$, this insures that $\EE[f(\alpha)] = \EE[\log(|\alpha|^2/(\Im\alpha))] >\EE[\log(|\alpha|^2)]=0$. For large $|u|$ we have
$\lim_{|u|\rightarrow\infty} \EE[f(u+\alpha)]-f(u+i) = \log(|\zeta|^2) >0$.
\end{proof}

We now return to the main argument. The upper half plane $\CC^+$ with the hyperbolic Riemannian metric will be denoted by $\HH$.
This metric on $\HH=\{z=x+iy: y>0\}$ is given by $ds^2= (dx^2+dy^2)/y^2$ and the unit speed geodesics are 
$$
\gamma(t)=\pmat{a&b\\c&d} \cdot ie^t = \frac{aie^t+b}{cie^t+d},
$$
where $a,b,c,d \in \RR$ and $ad-bc=1$. The induced length metric of negative curvature on $\HH$ is given by
\begin{equation*} %\label{eq-d_H}
d_\HH(z_1, z_2) = \arcosh\left( 1+ \frac{|z_1-z_2|^2}{2\Im(z_1)\Im(z_2)}\right) \;.
\end{equation*}

A function $f:\HH\rightarrow\RR$ is called geodesically convex if $f\circ\gamma$ is convex for every geodesic $\gamma$.
\begin{lemma} 
\label{lem-est1}
The function $f$ defined by \eqref{eq-def-f} is geodesically convex on $\HH$.
It is strictly convex on geodesics which do not have $0$ as a limit point.
\end{lemma}
\begin{proof} 
Let $ad-bc=1$ and
$
\gamma(t)={(aie^t+b)}/{(cie^t+d)}
$
be a unit speed geodesic in $\HH$. 
Then 
$$
f\circ\gamma (t)= -\log\left(\Im\; \frac{(ad-bc)ie^t-bd-ace^{2t}}{a^2e^{2t} +b^2}\right)=
\log(a^2e^{2t}+b^2)-t\;
$$
giving the second derivative
$$
(f\circ\gamma)''(t)=\frac{4a^2 b^2 e^{2t}}{(a^2e^{2t}+b^2)^2}\,\geq 0,
$$
which is strictly positive if $a\neq 0$ and $b\neq 0$.
Thus, $f\circ \gamma$ is convex and strictly convex if $a\neq 0, b\neq 0$.
If $a=0$ or $b=0$ then the geodesic $\gamma(t)$ approaches $0\in\CC$ for one of the limits
$t\to\pm \infty$.
\end{proof}

Next we need to introduce the concept of barycenter and we also introduce some notation for the hyperbolic midpoint. We refer to  the paper by Sturm \cite{Stu} for details and proofs.

\begin{defini}
\
\begin{enumerate}[{\rm(i)}]
\item Let $d_\HH$ be the distance induced by the hyperbolic Riemannian metric on $\HH$
 and let $X$ be a random variable on the probability space $(\Omega,\Aa,\PP)$
 with values in $\HH$ such that
 $ \EE[\, d_\HH(X,\beta)]$ exists for some $\beta\in\HH$. The set of such random variables is denoted by $L^1(\Omega,\HH)$.
 Then there exists a unique minimizer $b(X)\in\HH$ which minimizes, for any $\beta\in\HH$,
 $$
 \HH\ni z\mapsto \int\left(d_\HH(X(\omega),z)^2-d_\HH(X(\omega),\beta)^2\right)\,d\PP(\omega) =
 \EE \left[d_\HH(X,z)^2-d_\HH(X,\beta)^2\right].
 $$
 
The point $b(X)$ is called the barycenter (or more precisely $d^2$-barycenter) of $X$.
We will also call it the hyperbolic expectation value and denote it by $\EE_\HH (X)$, thus
 $$
 \EE_\HH(X)=\underset{z\in\HH}{{\rm minimizer}} \;\;\EE [d_\HH(X,z)^2-d_\HH(X,\beta)^2].
 $$
 
\item For two points $z_0, z_1 \in \HH$, let $\gamma:[0,1]\to\HH$ with $\gamma(0)=z_0$
and $\gamma(1)=z_1$ be the connecting geodesic. For any $\lambda \in [0,1]$ we define the hyperbolic affine combination
by
$$
\lambda z_0 \oplus (1-\lambda) z_1 := \gamma(\lambda)\;.$$ 
In particular, $({z_0\oplus z_1})/{2}$ is the midpoint of the joining geodesic, 
i.e., the unique point $z\in\HH$ such that $d_\HH(z,z_0)=d_\HH(z,z_1)=\frac12 d_\HH(z_0,z_1)$.\\

\item For a random variable $X\in L^1(\Omega,\HH)$ 
we construct an independent, identically distributed copy $X'$ 
and define the double-average barycenter by
$$
\EE_\HH^{(2)}[X] := \EE_\HH \left[ \frac12X\oplus \frac12X'\right]\;.
$$
\end{enumerate}
\end{defini}

Since we are assuming $Q$ is bounded, for any $z$ we find $\alpha(z,Q)\in L^1(\Omega,\HH)$. 
The following properties will be important.
%Proofs can be found in \cite[Section~6]{Stu} by K.~T.~Sturm about probability measures on spaces of nonpositive curvature.

\begin{enumerate}[(i)]
\item If the distribution of $X$ is supported within a geodesically convex set $A\subset \HH$, then the barycenter $\EE_\HH(X)$ 
lies inside $A$ \cite[Proposition 6.1]{Stu}. Using this twice one also finds that $\EE^{(2)}_\HH(X)$ lies in $A$.
\item For any geodesically convex function $h$ on $\HH$ one has Jensen's inequalities \cite[Theorem 6.2]{Stu}
$$
h(\EE_\HH[X]) \leq \EE \,[h(X)]\;\qtx{and}
h(\EE^{(2)}_\HH[X])\leq \EE \,[h(X)]\;.
$$
The second inequality follows from the first one applied twice.
\item Spaces of nonpositive curvature are doubly convex, i.e. the distance function on $\HH\times\HH$ is convex and
one finds (cf. \cite[Corollary 2.5; Theorem 6.3]{Stu}),
$$
d_\HH(\EE_\HH[X],\EE_\HH[Y])\,\leq\, \EE[\,d_\HH(X,Y)]\qtx{and}
d_\HH(\EE^{(2)}_\HH[X],\EE^{(2)}_\HH[Y])\,\leq\, \EE[\,d_\HH(X,Y)]\;.
$$
In particular, this implies that the map $z\mapsto \EE_\HH^{(2)}[\,\alpha(z,Q)]$ is continuous in $z$,
because if $z_n\to z$ in $\HH$, then one finds by continuity of $\alpha(\cdot,\cdot)$ and Dominated Convergence that 
$$
\lim_{n\to\infty} d_\HH\big(\EE_\HH^{(2)}[ \alpha(z_n,Q)], \EE_\HH^{(2)}[\alpha(z,Q]\big) \leq 
\lim_{n\to\infty} \EE[d_\HH(\alpha(z_n,Q),\alpha(z,Q))] = 0
$$

\end{enumerate}

Now we can show the crucial step.

\begin{lemma}\label{findzeta} If the variance of $Q$ is positive and $Q$ is bounded, then
there exists a $\zeta\in \HH$, $|\zeta|>1$ such that $\EE^{(2)}_\HH[\alpha(\zeta,Q)] = i$\;.
\end{lemma}
\begin{proof} 
Let us define the continuous function
$$
g(a,b):=\EE_\HH^{(2)} \left[\alpha(a+ib,Q)\right]=\frac1b \EE_\HH^{(2)}\left[ a+\frac{a}{a^2+b^2}-Q+i \frac{b}{a^2+b^2} \right]\,.
$$
Here we used that $z\mapsto z/b$ is an isometry on $\HH$. We need to find $(a,b)$ such that $g(a,b)=i$ with $a^2+b^2>1$. By using that $z\mapsto z+c$ for $c\in\RR$ is an isometry, one has
$$
b\, g(a,b) = a+\frac{a}{a^2+b^2} + \EE^{(2)}_\HH\left[-Q+i\frac{b}{a^2+b^2}\right]\;.
$$
We are assuming that $Q$ is almost surely bounded, so almost surely $|Q|<a_0$ for some $a_0$.
The set $\{z \in \HH : |\Re z|<a_0\}$ is geodesically convex in $\HH$,
therefore $\left|\Re\, \EE^{(2)}_\HH\left(-Q+i\frac{b}{a^2+b^2}\right)\, \right|<a_0$
and we find
\begin{equation}\label{111}
\re\,g(-a_0,b) < 0\qtx{and} \re\,g(a_0,b)>0\;.
\end{equation}

Also, the set $A_t:=\{z\in\HH\,:\,\Im(z)\ge t\,,|z|<\sqrt{a_0^2+t^2}\}$ (where the $|z|$ is the usual Euclidean norm) is geodesically convex in $\HH$. We have $\Im b g(a,b)= \Im \EE^{(2)}_\HH\left[-Q+i\frac{b}{a^2+b^2}\right]$ and the values of $-Q+i\frac{b}{a^2+b^2}$ lie in the set $A_t$ with $t:=b/(a^2+b^2)$. Hence $\EE^{(2)}_\HH\left[-Q+i\frac{b}{a^2+b^2}\right]$ also lies in $A_t$ which implies
$$
\Im b g(a,b)\,<\, \sqrt{a_0^2+\frac{b^2}{(a^2+b^2)^2}}\;.
$$
The right hand side is uniformly bounded for say $|b|>1$.
Hence there exists $b_0$ such that for all $a$
$$
\Im g(a,b_0) \,<\, 1
$$

Now similarly to Lemma~\ref{lem-est1} one finds that $z\mapsto \log(\Im(z))$ is geodesically concave. By Jensen's inequality we have
\begin{align*}
\log(\Im g(a,b)) &=
\log\left(\Im \EE_\HH \left[\frac12\alpha(a+ib,Q)\oplus\frac12\alpha(a+ib,Q')\right]\,\right) \notag \\
& \geq \EE\,\left[\frac12 \log \left(\frac{(Q-Q')^2}{4b^2}+\frac{1}{(a^2+b^2)^2} \right)\right], 
\end{align*}
since
\begin{equation*}
\frac12\alpha(a+ib,Q)\oplus\frac12\alpha(a+ib,Q')=\frac{a}{b} + \frac{a}{b(a^2+b^2)} -\frac{Q+Q'}{2} + i\sqrt{\frac{(Q-Q')^2}{4b^2} + \frac{1}{(a^2+b^2)^2}}.
\end{equation*}
As the variance of $Q$ is positive, $\EE[(Q-Q')^2]>0$. Hence there exists $\epsilon>0$ and $1\geq p>0$ such that  
$\PP[(Q-Q')^2\geq \epsilon^2]=p$. This implies
$$
\log(\Im g(a,b))\,\geq\, \frac{(1-p)}{2}\, \log\left(\frac{1}{(a^2+b^2)^2}\right) + \frac{p}{2} \log 
\left(\frac {\epsilon^2}{4b^2} + \frac{1}{(a^2+b^2)^2}\right)\;.
$$
If $a^2+b^2\le 1$, the right hand side is clearly bigger than $0$. Hence the region where the right hand side is $>0$ includes an open neighborhood of the unit circle. Moreover, for any $a$, letting $b\to 0$, the right hand side approaches infinity. In case $a\neq 0$ one can choose $b < \min\big\{|a|,\frac12 \epsilon (a^2)^{\frac{p-1}{p}} \big\} $ to get $\log(\Im g(a,b)) > 0 $. Hence, there is a continuous function $b_1$ defined on $[-a_0,a_0]$ such that
$$
\Im(g(a,b_1(a)))>1,\; \quad a^2+b_1(a)^2\geq 1.
$$
One clearly may assume $b_1(a)<b_0$. 

Let $\Gamma$ be the closed path in $\HH$ given by  $(a,b_1(a)),\, -a_0\leq a \leq a_0$, followed by $(a_0,b),\, b_1(a_0)\leq b \leq b_0$ followed by $(a,b_0),\, a_0 \geq a \geq -a_0$ and followed by $(-a_0,b),\, b_0 \geq b \geq b_1(-a_0)$.
Then, using \eqref{111}, $g(\Gamma)$ is a  path enclosing $i$. As $\Gamma$ is null-homotopic, its image is null-homotopic in $g(\HH)$, hence
there exists $(a,b)$ inside $\Gamma$ such that $g(a,b)=i$. As $(a,b)$ lies inside $\Gamma$, we have $a^2+b^2>1$.
Setting $\zeta=a+ib$, this finishes the proof.
\end{proof}

\begin{proof}[Proof of Theorem \ref{keyestthm}:]
We need to establish \eqref{eq-est-f}. Let $X$ and $X'$ be independent and identically distributed as $u+\alpha(\zeta,Q)$, where $u\in\RR$ and $\zeta$ is chosen according to Lemma \ref{findzeta}. Clearly, $\EE_\HH^{(2)}[X]=u+i$ since $z\mapsto u + z$ is an isometry on $\HH$ for $u\in\RR$. By convexity of $f$ and Lemma \ref{lem-est1}
\begin{align*}
f(u+i) & = f\left(\EE_\HH\left[ \frac12 X \oplus \frac12 X' \right]\right)\,\leq\,
\EE f \left( \frac12 X\oplus \frac12 X' \right) \notag \\
& < \EE \left[\frac12 f(X) + \frac12 f(X') \right]
= \EE f(u+\alpha(\zeta,Q)),
\end{align*}
which is the desired inequality. 
As the essential support of $Q$ contains more than two points, one has with positive
probability that $X$ and $X'$ are not on a geodesic with $0$ as a limit point giving the strict inequality. In order to obtain the uniform statement in \eqref{eq-est-f} we compute $\lim_{|u|\rightarrow\infty} \EE[f(u+\alpha)]-f(u+i) = \log(|\zeta|^2) >0$ as before. 
\end{proof}

\section{Proof of Theorem \ref{mainthm} and Theorem \ref{largekappathm}}\label{sectionproof2}

\begin{proof}[Proof of Theorem \ref{mainthm}:]
We will use the notation $T_{\kappa,z,s}$ to indicate the dependence of the operator $T$ on its parameters. Given $E_0\in\RR$ we use Lemma \ref{onedimlemma} to find a bounded open interval $I$ containing $E_0$, $s\in(0,1/2)$, $m\in\NN$ and $\delta<1$ so that
\begin{equation*}
\sup_{E\in I} \|T_{0,E,s}^m 1\|_{\infty,\overline\RR} < 1.
\end{equation*}
Here the second subscript on the sup norm indicates the set over which the sup is taken. Using the continuity in $\kappa$ given by Lemma~\ref{strong continuity of T to m in kappa} we can find $\kappa_0$ such that for $0\le\kappa\le\kappa_0$,
\begin{equation*}
\sup_{E\in I} \|T_{\kappa,E,s}^m 1\|_{\infty,\overline\RR} < 1.
\end{equation*}
Next we invoke Lemma \ref{prelim properties of T}\eqref{4} to replace the sup over $\overline\RR$ with a sup over $\overline{\CC^+}$. This gives, for $0\le\kappa\le\kappa_0$,
\begin{equation} \label{Tmest0}
\sup_{E\in I} \|T_{\kappa,E,s}^m 1\|_{\infty,\overline{\CC^+}} < 1.
\end{equation}
Lemma~\ref{prelim properties of T}\eqref{2b} (applied to the one-element family $\{f_{\alpha}\} = \{1\}$) now implies the existence of an $\epsilon_0>0$ such that
\begin{equation}\label{Tmest}
\sup_{0\le\epsilon<\epsilon_0} \sup_{E\in I}\|T_{\kappa,E+i\epsilon,s}^m 1\|_{\infty,\overline{\CC^+}} =:\delta < 1
\end{equation}
for $0\le \kappa \le \kappa_0$.

Given $n\in\NN$ we write $n= am + b$ for non-negative integers $a,b$ with $b<m$ and $a>n/m-1$. The linearity of $T_{\kappa,E,s}$ together with the monotonicity property Lemma \ref{prelim properties of T}\eqref{3} and \eqref{Tmest} imply
\begin{equation*}
\sup_{0\le\epsilon<\epsilon_0} \sup_{E\in I}\|T_{\kappa,E+i\epsilon,s}^{am} 1\|_{\infty,\overline{\CC^+}} \le \delta^a \le \delta^{n/m-1}.
\end{equation*}
So
\begin{align*}
\sup_{0\le\epsilon<\epsilon_0} \sup_{E\in I}\|T_{\kappa,E+i\epsilon,s}^{n} 1\|_{\infty,\overline{\CC^+}} 
&\le \sup_{0\le\epsilon<\epsilon_0} \sup_{E\in I}\|T_{\kappa,E+i\epsilon,s}\|^b\delta^{n/m-1}\\
&\le C \delta^{n/m},
\end{align*}
where we used the uniform bound on $\|T_{\kappa,E+i\epsilon,s}\|$ given by Lemma \ref{prelim properties of T}\eqref{1}. The theorem now follows from \eqref{dynamical system bound} with $\ell=\delta^{-1/m}$.
\end{proof}

We now turn to the {\it proof of Theorem \ref{largekappathm}}. Let $\sigma_0$ and $\sigma_1$ be the densities of the marginal measures of $\sigma$.

\begin{lemma}[Large coupling] \label{large coupling} Under Assumptions \ref{assumptions2}, for every $s\in (0,1)$ there exists a $\kappa_0<\infty$, depending on $\|\nu\|_\infty$, $\|\sigma_{0}\|_\infty$, $\|\sigma_1\|_\infty$ and $s$, such that for $\kappa\ge\kappa_0$,
\begin{equation} \label{largedisorderbound}
\sup_{E\in\RR} \|T_{\kappa, E,s}\|_\infty <1.
\end{equation}
\end{lemma}

\begin{proof}  Following the proof of Lemma~\ref{prelim properties of T}(\ref{1}) (see Section~\ref{appendix} below), we need bounds for $\EE[|w-q_i|^{-s}|]$ for $i=0,1$, uniform in $w$. By the definition of the marginals of $\sigma$, such bounds are
\begin{eqnarray*} \iint \frac{\nu(r) \sigma_i(p)\,dr dp_i}{|w-r-\kappa p_i|^s}
&\le& \|\nu\|_\infty \, \|\sigma_i\|_\infty\, \kappa^{-s} \int_{-K}^K dr\,  \int_{-1}^1 \frac{dp_i}{\big|(w-r)/\kappa- p_i|^s}
\\
&\le& \|\nu\|_\infty \, \|\sigma_i\|_\infty\, \kappa^{-s}\,\frac2{1-s} \int_{-K}^K dr
\\
&\le& 4K\,\frac{\|\nu\|_\infty \, \|\sigma_i\|_\infty}{(1-s)} \, \kappa^{-s}.
\end{eqnarray*}
By choosing $\kappa$ large enough this leads to (\ref{largedisorderbound}). 
\end{proof}

Given this result we can follow the steps above, except with $m=1$, to obtain the required fractional moment estimate as in Theorem~\ref{mainthm}, now in the large coupling regime, and thus, as before, conclude pure point spectrum to prove Theorem~\ref{largekappathm}.

\section{Proof of Lemma~\ref{prelim properties of T} and Lemma~\ref{strong continuity of T to m in kappa}} \label{appendix}

In this final section we provide the proofs of the technical Lemmas~\ref{prelim properties of T} and \ref{strong continuity of T to m in kappa}, which have been used in the proof of Theorem~\ref{mainthm} above. 

\begin{proof}[Proof of Lemma~\ref{prelim properties of T}:] \eqref{1} The argument that lead to the bound (\ref{initial bound}) of $\EE[|g_0|^s]$ can be employed again to show well-posedness and boundedness of the operator $T_{\kappa,z,s}$, uniformly in $\kappa$ and $z$. We have
\begin{eqnarray} \label{double bound}
\lefteqn{\sup_{w\in \overline{\CC^+}}|(T_{\kappa,z,s}f)(w)|}\notag
\\ 
&\le& 2 \|f\|_\infty \sup_{w\in \overline{\CC^+}} \EE\Big[\Big|\frac{-1}{2w + \sqrt{2}(z-q_0)}\Big|^s + \Big|\frac{-1}{2w + \sqrt{2}(z-q_1)}\Big|^s\Big]\notag
\\
&\le&2^{1-s/2} \|f\|_\infty \sup_{w\in \overline{\CC^+}} \EE \Big[\frac1{|(\sqrt{2}\Re(w)+\Re(z)-q_0|^{s}} + \frac1{|(\sqrt{2}\Re(w)+\Re(z)-q_1|^{s}}\Big]\notag
\\
&\le&2^{1-s/2} \|f\|_\infty \sup_{w \in \RR} \EE \Big[\frac1{|w-q_0|^{s}} + \frac1{|w-q_1|^{s}}\Big]
\end{eqnarray}
In order to bound this we use
\begin{eqnarray}
\EE \Big[ \frac{1}{|w-q_0|^s} \Big] & = & \int_{\RR^2} \int_{-K}^K \frac{1}{|w-r-\kappa p_0|^s} \nu(r)\,dr\,d\sigma(p_0,p_1) \notag
\\
& \le & \frac{2\|\nu\|_{\infty} K^{1-s}}{1-s},
\end{eqnarray}
where, after bounding the $r$-integral similar to (\ref{initial bound}), the $\sigma$-integral becomes trivial. The same bound holds for the other term in (\ref{double bound}), leading to (\ref{uniform operator norm bound}).

Thus we have proved that $\|T_{\kappa,z,s}f\|_\infty$ is finite and that $\|T_{\kappa,z,s}\|$ is bounded as in \eqref{uniform operator norm bound}. That $T_{\kappa,z,s}f$ is also a continuous function on $\overline{\CC^+}$ is equivalent to the map $z\mapsto T_{\kappa,z,s}f$ being continuous since $z$ and $w$ always appear in the combination $w+z/\sqrt{2}$. This continuity is shown next. Also $\lim_{|w|\to\infty} (T_{\kappa,z,s}f)(w)=0$, proving that $T_{\kappa,z,s}f \in C(\overline{\CC^+})$.

\eqref{2} Fix $f\in C(\overline{\CC^+})$. We prove the continuity of the map $z\mapsto T_{\kappa,z,s}f$. For fixed $z\in\CC^+\cup\RR$ and a sequence $z_n$ in $\CC^+\cup\RR$ with $|z-z_n|\to0$ as $n\to\infty$, we have (we abbreviate $\EE:=\EE_q$ and $\phi^{\pm}_z:=\phi^\pm_{z,q}$)
\begin{eqnarray}
\lefteqn{\sup_{w\in\overline{\CC^+}} |(T_{\kappa,z_n,s}f - T_{\kappa,z,s}f)(w)|}\notag\\
&=& \sup_{w\in \overline{\CC^+}} \left|\EE \Big[f(\phi^+_{z_n}(w))\big(|\phi^+_{z_n}(w)|^s + |\phi^-_{z_n}(w)|^s\big) - f(\phi^+_z(w))\big(|\phi^+_z(w)|^s + |\phi^-_z(w)|^s\big)\Big]\right|\notag
\\
&\le&\sup_{w\in \overline{\CC^+}} \left|\EE \Big[f(\phi^+_{z_n}(w))\big(|\phi^+_{z_n}(w)|^s - |\phi^+_z(w)|^s+ |\phi^-_{z_n}(w)|^s - |\phi^-_z(w)|^s\big)\right|\notag
\\
& &\mbox{} + \sup_{w\in \overline{\CC^+}} \left|\EE \Big[\big(f(\phi^+_{z_n}(w)) - f(\phi^+_z(w))\big)\big(|\phi^+_z(w)|^s + |\phi^-_z(w)|^s\big)\Big]\right|\notag
\\
&\le&\|f\|_\infty \sup_{w\in \overline{\CC^+}} \EE \Big[\big||\phi^+_{z_n}(w)|^s - |\phi^+_z(w)|^s\big| + \big||\phi^-_{z_n}(w)|^s - |\phi^-_z(w)|^s\big|\Big] \notag
\\
& & \mbox{} + \sup_{w\in \overline{\CC^+}}\EE\Big[\big|f(\phi^+_{z_n}(w)) - f(\phi^+_z(w))\big|\big(|\phi^+_z(w)|^s + |\phi^-_z(w)|^s\big)\Big]. \label{twotermsplit}
\end{eqnarray}
Writing 
$$ \phi^\pm_{z_n}(w) - \phi^\pm_z(w) = \frac{z_n-z}{\sqrt{2}} \left( \frac{1}{(\sqrt{2}w+z_n -q_0)(\sqrt{2}w+z-q_0)} \pm \frac{1}{(\sqrt{2}w+z_n -q_1)(\sqrt{2}w+z-q_1)} \right)
$$
we arrive at the estimate
\begin{eqnarray} \label{2sestimate}
\lefteqn{\EE\Big[\big||\phi^\pm_{z_n}(w)|^s - |\phi^\pm_z(w)|^s\big|\Big]} \notag
\\
&\le& 2^{-s/2} |z_n-z|^s \,\EE\Big[\frac1{|\sqrt{2}w+z_n-q_0|^{s}\, |\sqrt{2}w+z-q_0|^{s}} + \frac1{|\sqrt{2}w+z_n-q_1|^{s}\, |\sqrt{2}w+z-q_1|^{s}}\Big] \notag
\\
&\le& 2^{-s/2} |z_n-z|^s \,\EE\Big[\frac1{|\Re(\sqrt{2}w+z_n)-q_0|^{s}\,|\Re(\sqrt{2}w+z)-q_0|^{s}} \notag
\\
& & \mbox{} + \frac1{|\Re(\sqrt{2}w+z_n)-q_1|^s\,|\Re(\sqrt{2}w+z)-q_1|^{s}} \Big] \notag
\\
&\le& 2^{-s/2} |z_n-z|^s \,\EE\Big[\frac1{|\Re(\sqrt{2}w+z_n)-q_0|^{2s}} + \frac1{|\Re(\sqrt{2}w+z)-q_0|^{2s}} \Big] \notag
\\
&+& 2^{-s/2} |z_n-z|^s \,\EE\Big[\frac1{|\Re(\sqrt{2}w+z_n)-q_1|^{2s}} + \frac1{|\Re(\sqrt{2}w+z)-q_1|\big)^{2s}} \Big].
\end{eqnarray}
When taking the supremum over $w$ we can set $z_n=0$ and $z=0$, respectively. By using the bound that led to \eqref{initial bound} we see that 
\beq \label{1sttermbound}
\sup_{w\in \overline{\CC^+}} \EE\Big[\big||\phi^\pm_{z_n}(w)|^s - |\phi^\pm_z(w)|^s\big|\Big] \le \frac{8 \cdot 2^{-s/2} K^{1-2s} \|\nu\|_{\infty}}{1-2s}\, |z_n-z|^s
\ee
which vanishes as $n\to\infty$.

Now, we come to the second term,
$$
\EE\Big[\big|f(\phi^+_{z_n,q}(w)) - f(\phi^+_{z,q}(w))\big| \,|\phi^\pm_{z,q}(w)|^s \Big].
$$
We have reintroduced $q$ in the notation and we use
$$ 
|\phi^\pm_{z,q}(w)|^s \le 2^{-s/2} \left( \frac1{|\Re(\sqrt{2}w+z)-q_0|^s} + \frac1{|\Re(\sqrt{2}w+z)-q_1|^s} \right).
$$
The expectation of $q\mapsto f(\phi^+_{z_n,q}(w))\,|\phi^\pm_{z,q}(w)|^s$ is bounded uniformly in $n$ by the previous arguments, and $f(\phi^+_{z_n,q}(w))$ converges point-wise to $f(\phi^+_{z,q}(w))$ as $n\to\infty$, almost surely with respect to the measure $\mu$. An application of dominated convergence finishes the proof of (\ref{2}).

\eqref{2b} We proceed inductively in $m$. For $m=1$ we split $\|(T_{\kappa,E+i\epsilon,s} - T_{\kappa,E,s})f_{\alpha}\|_{\infty}$ into two terms analogous to (\ref{twotermsplit}). The first term can be treated as the the proof of part (\ref{2}), using uniform boundedness of $\|f_{\alpha}\|_\infty$ and that the analogue of the bound (\ref{1sttermbound}) has the required uniformity properties.

The second term consists of two contributions (one for each sign in $\phi^{\pm}$),
\beq \label{equicontterm}
\EE_q \left[ |f_{\alpha}(\phi_{E+i\epsilon,q}^+(w) - f_{\alpha}(\phi_{E,q}^+(w))| |\phi_{E,q}^{\pm}(w)|^s \right].
\ee
It is here where one has to go beyond the $\hbox{`soft'}$ dominated convergence argument used in the proof of part (\ref{2}). With $0<\beta < 1/2$ we split the $r$-integration within $\EE_q$ into the regions
\begin{eqnarray*}
A & = & \{r: |\sqrt{2}w+E-r-\kappa p_i| \ge \epsilon^{\beta} \:\mbox{for}\: i=0 \:\mbox{and} \: i=1\} 
\\
B & = & \RR \setminus A.
\end{eqnarray*}
In region $A$ we have $|\phi_{E+i\epsilon,q}^+(w) - \phi_{E,q}^+(w)| \le \sqrt{2} \epsilon^{1-2\beta}$. Using the equi-continuity of $\{f_{\alpha}\}$ we find that for every $\delta>0$ there exists $\epsilon_0>0$ such that $|f_{\alpha}(\phi_{E+i\epsilon,q}^+(w) - f_{\alpha}(\phi_{E,q}^+(w))| < \delta$ for $r\in A$ and $0\le \epsilon \le \epsilon_0$. Thus the corresponding contribution to (\ref{equicontterm}) satisfies
\[
\int_{\RR^2} \int_{A} \ldots \le C(K,s) \delta
\]
uniformly in $\kappa$, $E$ and $w$. Moreover, boundedness of $\{f_{\alpha}\}$ and that $B$ consists of two intervals of length $2\epsilon^{\beta}$ gives the bound
\[
\int_{\RR^2} \int_{B} \ldots \le C(s) \|\nu\|_{\infty} \epsilon^{\beta(1-s)},
\]
also uniformly in $\kappa$, $E$ and $w$. These bounds combine to prove (\ref{unicontbound}) for $m=1$.

To carry out the inductive step we assume that (\ref{unicontbound}) has been proved for all integers up to $m-1$ and write
\begin{eqnarray}
\lefteqn{\| (T_{\kappa,E+i\epsilon,s}^m - T_{\kappa,E,s}^m)f_{\alpha}\|_{\infty}} \notag
\\
& \le & \| T_{\kappa,E+i\epsilon,s}^{m-1} (T_{\kappa,E+i\epsilon,s} - T_{\kappa,E,s}) f_{\alpha} \|_{\infty} + \| (T_{\kappa,E+i\epsilon,s}^{m-1} - T_{\kappa,E,s}^{m-1}) T_{\kappa,E,s} f_{\alpha} \|_{\infty} \notag
\\
& \le & C^{m-1} \|(T_{\kappa,E+i\epsilon,s}- T_{\kappa,E,s}) f_{\alpha}\|_{\infty} + \| (T_{\kappa,E+i\epsilon,s}^{m-1} - T_{\kappa,E,s}^{m-1}) T_{\kappa,E,s} f_{\alpha} \|_{\infty}, \label{almostdone}
\end{eqnarray}
where Lemma~\ref{prelim properties of T}(i) was used. That the first term in (\ref{almostdone}) goes to zero uniformly in $\kappa$, $E$ and $\alpha$ is the case $m=1$. The family $\{T_{\kappa,E,s} f_{\alpha}\}_{\kappa,E,\alpha}$ is equi-continuous by Lemma~\ref{lemmaequicont} below. Thus the second term goes to zero uniformly in $\kappa$, $E$ and $\alpha$ by the inductive assumption.

\eqref{3} is clearly true.

\eqref{4} We use the fact that for any holomorphic function $\phi$, the function $w\mapsto |\phi(w)|^s$ (for any power $s\ge0$) is a subharmonic function on $\CC^+$. Recall that a function $f:\CC^+\to\RR\cup\{-\infty\}$ is called submean if the average of $f$ on any circle (inside $\CC^+$) is larger than (or equal to) the value of $f$ at the center of this circle. A function $f$ is called upper-semicontinuous if for any sequence $z_n$ converging to $z$ we have $\limsup f(z_n)\le f(z)$. A function is called subharmonic if it is both submean and upper-semicontinuous. Subharmonicity is preserved if we add or integrate subharmonic functions. The maps $w\mapsto \phi^\pm_{z,q}$ are both holomorphic. Hence, the integrand of $T_{\kappa,z,s} 1$ is a subharmonic function of $w$, and so is $w\mapsto (T_{\kappa,z,s} 1)(w)$. The same applies to the map $w\mapsto (T_{\kappa,z,s}^m 1)(w)$ if we expand the $m$th power as a multiple integral. 

An important property of subharmonic functions is that the supremum of a subharmonic function on a domain with compact closure is taken at the boundary of the domain, which in our case is $\RR\cup\{i\infty\}$. 
%Since, as we proved above, $T_z f$ (and also its powers) is continuous up to the boundary we could omit any single point and replace the supremum over $\RR\cup\{i\infty\}$ by a supremum over $\RR$. Altogether, 
This implies \eqref{sup over R}.
\end{proof}

In the proof of Lemma~\ref{prelim properties of T}(\ref{2b}) above we have used the following

\begin{lemma} \label{lemmaequicont}
Let $\{f_{\alpha}: \alpha \in {\mathcal J}\}$ be an equi-continuous and bounded subset of $C(\overline{\CC^+})$ and $s\in(0,1/2)$. Then
\beq \label{equicontset}
\{ T_{\kappa, E, s}f_{\alpha}: \kappa \ge 0, E\in \RR, \alpha \in {\mathcal J} \}
\ee
is equi-continuous and bounded in $C(\overline{\CC^+})$.
\end{lemma}

\begin{proof} We have
\begin{eqnarray}
(T_{\kappa,E,s} f_{\alpha})(w) - (T_{\kappa,E,s} f_{\alpha})(w') & = & \EE_q \Big[ f_{\alpha}(\phi_{E,q}^+(w)) (|\phi_{E,q}^+(w)|^s - |\phi_{E,q}^+(w')|^s) \Big] \label{1stterm}
\\
& & \mbox{} + \EE_q \Big[ \left(f_{\alpha}(\phi_{E,q}^+(w)) - f_{\alpha}(\phi_{E,q}^+(w')) \right) |\phi_{E,q}^+(w')|^s \Big] \label{2ndterm}
\\
& & \mbox + \;\;\mbox{similar terms with $|\phi^+|^s$ replaced by $|\phi^-|^s$}. \notag
\end{eqnarray}
We focus on the terms involving $|\phi^+|^s$, with the $|\phi^-|^s$-terms giving the same bounds. The absolute value of (\ref{1stterm}) is bounded by
\[
\|f_{\alpha}\|_{\infty} \left| \EE_q (|\phi_{E,q}^+(w)|^s - |\phi_{E,q}^+(w')|^s) \right|.
\]
By assumption, $\|f_{\alpha}\|_{\infty} \le C$. Moreover,
\begin{eqnarray*}
\lefteqn{\left| \EE_q (|\phi_{E,q}^+(w)|^s - |\phi_{E,q}^+(w')|^s) \right|} 
\\
& \le & C(s) \int_{\RR^2} \int_{-K}^K \left| \frac{1}{\sqrt{2}w + E-r-\kappa p_0} - \frac{1}{\sqrt{2}w' + E-r-\kappa p_0} \right|^s \,\nu(r)\,dr\,d\sigma(p_0,p_1)
\\
& & \mbox{} + \;\;\mbox{a similar term with $p_0$ replaced by $p_1$}.
\end{eqnarray*}
The appearing $r$-integral, using $s<1/2$ and the elementary bound $|\frac{1}{a} - \frac{1}{b}|^s \le |b-a|^s \left( \frac{1}{|a|^{2s}} + \frac{1}{|b|^{2s}} \right)$ (similar to what was done in (\ref{2sestimate}) above), can be seen to satisfy a bound of the form $C(K,s) |w-w'|^s$, thus making the $\sigma$-integration trivial. This shows the required equi-continuity of the set (\ref{equicontset}) for the contribution by the term (\ref{1stterm}).

The term (\ref{2ndterm}) is treated in a way similar to how we bounded (\ref{equicontterm}) above. Choose $0< \beta < 1/2$ and, for fixed $q = \kappa (p_0,p_1)$, split the $r$-integral within $\EE_q$ into two regions,
\begin{eqnarray*}
A' & = & \{r: |\sqrt{2}w+E-r-\kappa p_i| \ge |w-w'|^{\beta} \: \mbox{and} \: |\sqrt{2}w'+E-r-\kappa p_i| \ge |w-w'|^{\beta} 
\\
& & \mbox{for} \: i=0 \: \mbox{and} \: i=1\}, \\
B' & = & \RR \setminus A'.
\end{eqnarray*}
Correspondingly, we write
\beq
\EE_q \Big[ \left(f_{\alpha}(\phi_{E,q}^+(w)) - f_{\alpha}(\phi_{E,q}^+(w')) \right) |\phi_{E,q}^+(w')|^s \Big] = \int_{\RR^2} \int_{A'} \ldots + \int_{\RR^2} \int_{B' }\ldots
\ee
For $r\in A'$ one checks
\[
|\phi_{E,q}^+(w) - \phi_{E,q}^+(w')| \le 2|w'-w|^{1-2\beta}.
\]
Combined with the equi-continuity of $\{f_{\alpha}\}$ this means that for every $\delta>0$ there is $\delta'>0$ such that
\[
|w'-w| < \delta' \: \Longrightarrow \: |f_{\alpha}(\phi_{E,q}^+(w)) - f_{\alpha}(\phi_{E,q}^+(w'))| < \delta,
\]
which in turn gives
\beq \label{Abound}
\int_{\RR^2} \int_{A'} \ldots \le C(K,s) \|\nu\|_{\infty} \,\delta
\ee
for all $E$, $\kappa$ and $\alpha$.

Note that $B'$ is the union of four intervals of length $2|w'-w|^{\beta}$. From this one gets the bound
\begin{eqnarray} \label{Bbound}
\int_{\RR^2} \int_{B'} \ldots & \le & C(s) \|f_{\alpha}\|_{\infty} \|\nu\|_{\infty} |w-w'|^{\beta(1-s)} \notag \\
& \le & C'(s) \|\nu\|_{\infty} |w-w'|^{\beta(1-s)},
\end{eqnarray}
for all $E$, $\kappa$ and $\alpha$, having used boundedness of $\{f_{\alpha}\}$.

Combining (\ref{Abound}) and (\ref{Bbound}) gives the required equi-continuity property for the contribution by the term (\ref{2ndterm}) and thus completes the proof of Lemma~\ref{lemmaequicont}. 
\end{proof}

We finally turn to the 

\begin{proof}[Proof of Lemma~\ref{strong continuity of T to m in kappa}:] We proceed inductively and first show the claim for $m=1$, i.e.\ that
\beq \label{m1}
\lim_{\kappa\downarrow0} \sup_{E\in\RR} \big\|T_{\kappa,E,s}f - T_{0,E,s}f \big\|_\infty = 0.
\ee

Towards this, let us abbreviate $\phi^\pm_{E,r+\kappa p} := \phi^\pm_{E,(r+\kappa p_0,r+\kappa p_1)}, \phi^\pm_{E,r} := \phi^\pm_{E,(r,r)}$ and $\phi^\pm_{r+\kappa p} := \phi^\pm_{0,r+\kappa p}$. Then, using the triangle inequality we get,
\begin{eqnarray}
\lefteqn{\sup_{E\in\RR} \big\|T_{\kappa,E,s}f - T_{0,E,s}f \big\|_\infty}\notag
\\
&=&\sup_{E\in\RR} \sup_{w\in\overline{\RR}} \left|\iiint \Big[f(\phi^+_{E,r+\kappa p}(w))\big(|\phi^+_{E,r+\kappa p}(w)|^s + |\phi^-_{E,r+\kappa p}(w)|^s\big) \right. \notag
\\
&&\left.\phantom{\sup_{E\in\RR} \sup_{w\in\overline{\RR}} \iiint\;} - f(\phi^+_{E,r}(w)) |\phi^+_{E,r}(w)|^s\Big]\,\,\nu(r)dr\,d\sigma(p)\right|\notag
\\
&\le&\sup_{w\in\overline{\RR}} \iiint \big|f(\phi^+_{r+\kappa p}(w)) - f(\phi^+_{r}(w))\big|\,\big|\phi^+_{r}(w)|^s\,\nu(r)dr\,d\sigma(p) \label{first term}
\\
&+&\sup_{w\in\overline{\RR}} \iiint \big|f(\phi^+_{r+\kappa p}(w))\big|\,\big[|\phi^+_{r+\kappa p}(w)|^s - |\phi^+_{r}(w)|^s\big]\,\nu(r)dr\,d\sigma(p)\label{second term}
\\
&+&\sup_{w\in\overline{\RR}} \iiint \big|f(\phi^+_{r+\kappa p}(w))\big|\,|\phi^-_{r+\kappa p}(w)|^s\,\nu(r)dr\,d\sigma(p) .\label{third term}
\end{eqnarray}
Since the sup over $E$ is gone once we have taken the sup over $w$ we have set $E$ to $0$. We show now that each of the three terms goes to $0$ as $\kappa\downarrow0$. 

The first term \eqref{first term} is the most complicated one. We have the bound
\beq \label{fristtermbound}
|\phi_{r+\kappa p}^+(w) - \phi_r^+(w)| \le \frac{\kappa}{\sqrt{2}|\sqrt{2}w-r|} \left( \frac{1}{|\sqrt{2}w-r-\kappa p_0|} + \frac{1}{|\sqrt{2} w-r-\kappa p_1|} \right).
\ee
We consider first the region where $|\sqrt{2}w-r|>\kappa^a$ with $0<a<1/2$. Since $|p_i|\le1$ by Assumption~\ref{assumptions1} we have $|\sqrt{2}w-r-\kappa p_i|\ge |\sqrt{2}w-r|-\kappa \ge \kappa^a-\kappa\ge \frac12 \kappa^a$ for small $\kappa$. Therefore,
\begin{equation*}
\sup_{w\,:\,|\sqrt{2}w-r|>\kappa^a}\big|\phi^+_{r+\kappa p}(w) - \phi^+_{r}(w)\big| \;\le\;2\sqrt{2}\kappa^{1-2a}. 
\end{equation*}
Applying this to the (uniformly) continuous function $f$ we get
\begin{equation*}
\sup_{w\,:\,|\sqrt{2}w-r|>\kappa^a}\big|f(\phi^+_{r+\kappa p}(w)) - f(\phi^+_{r}(w))\big| \;\le\; \sup_{x,y\,:\,|x-y|\le 2\sqrt{2}\kappa^{1-2a}} |f(x)-f(y)| \;=\; o_\kappa(1),
\end{equation*}
where $o_\kappa(1)$ means that the term goes to $0$ as $\kappa\to0$. %[If $f\in C(\RR)$ then this property of uniform continuity has to be requested. If $f\in C(\overline{\RR})$ then, as a continuous function on a compact set, $f$ is automatically uniformly continuous.] 
In the region $|\sqrt{2}w-r|\le \kappa^a$ we use the integrability of $r\mapsto|\phi^+_r(w)|^s \nu(r)$. Altogether, we have the bound
\begin{eqnarray*}
\lefteqn{\iiint \big|f(\phi^+_{r+\kappa p}(w)) - f(\phi^+_{r}(w))\big|\,\big|\phi^+_{r}(w)|^s\big|\,\nu(r)dr\,d\sigma(p)}
\\
&\le&o_\kappa(1) \int_{|\sqrt{2}w-r|>\kappa^a} |\phi^+_r(w)|^s \,\nu(r)dr + 2\|f\|_\infty \int_{|\sqrt{2}w-r|\le\kappa^a} |\phi^+_r(w)|^s \,\nu(r)dr
\\
&\le&o_\kappa(1)\, \|\nu\|_\infty \int_{-K}^K \frac{dr}{|r|^s} + 2\|f\|_\infty \|\nu\|_{\infty} \int_{|r|\le \kappa^a} \frac{dr}{|r|^s}
\\
&=&o_\kappa(1) \,\frac{2\|\nu\|_\infty}{1-s}\, K^{1-s} + \kappa^{a(1-s)}\, \frac{4\|f\|_\infty \|\nu\|_{\infty}\|}{1-s}.
\end{eqnarray*}

In the second term \eqref{second term}, we get from (\ref{fristtermbound}),
\begin{eqnarray*} 
\lefteqn{\Big||\phi^+_{r+\kappa p}(w)|^s - |\phi^+_{r}(w)|^s\Big| \: \le \: \big| \phi_{r+\kappa p}^+(w) - \phi_r^+(w) \big|^s}
\\
&\le& 2^{-s/2} \kappa^s \frac{1}{|\sqrt{2}w-r|^s}\Big(\frac1{|\sqrt{2}w-r-\kappa p_0|^s} + \frac1{|\sqrt{2}w-r-\kappa p_1|^s}\Big)
\\
&\le&2^{-s/2} \kappa^s\Big(\frac{1}{|\sqrt{2}w-r|^{2s}} + \frac1{\sqrt{2}|w-r-\kappa p_0|^{2s}} + \frac1{|\sqrt{2}w-r-\kappa p_1|^{2s}}\Big).
\end{eqnarray*}
With this estimate we perform the $r$-integral in (\ref{second term}) to obtain the upper bound 
\beq \label{2nd term bound}
\frac{6\cdot 2^{-s/2}}{1-2s} \kappa^s K^{1-2s} \|\nu\|_{\infty} \|f\|_{\infty}.
\ee

The third term \eqref{third term} is analogous to the previous one, so that
\begin{eqnarray*}
\lefteqn{|\phi^-_{r+\kappa p}(w)|^s = 2^{-s/2} \Big|\frac1{\sqrt{2}w-r-\kappa p_0} - \frac1{\sqrt{2}w-r-\kappa p_1}\Big|^s}
\\
&\le&2^{-s/2} \kappa^s \,|p_0-p_1|^s\,\Big(\frac1{|\sqrt{2}w-r-\kappa p_0|^{2s}} + \frac1{|\sqrt{2}w-r-\kappa p_1|^{2s}}\Big).
\end{eqnarray*}
Then integrate this using $|p_0-p_1|\le 2$ and get a bound similar to (\ref{2nd term bound}). This finishes the proof of (\ref{m1}).

For $m\ge2$ we write, leaving out the indices $E$ and $s$,
$$ 
T_\kappa^{m} - T_0^{m} = T_\kappa^{m-1} (T_\kappa - T_0) + ( T_\kappa^{m-1} - T_0^{m-1}) T_0
$$
Since $\| T_\kappa^{m-1}\|_\infty\le \|T_\kappa\|_\infty^{m-1}$ is bounded uniformly in $\kappa$ by (\ref{uniform operator norm bound}),
\begin{eqnarray*}
\| T_\kappa^{m}f - T_0^{m} f\|_\infty &\le& \|T_\kappa^{m-1}\|_\infty \, \|T_\kappa f- T_0f\|_\infty +  \|(T_\kappa^{m-1} - T_0^{m-1}) (T_0 f)\|_\infty
\end{eqnarray*}
goes to $0$ as $\kappa\to0$ by assuming the induction hypothesis for the functions $f$ and $T_0f$.

\end{proof}


\begin{thebibliography}{AAA}

\bibitem[A]{A} M. Aizenman: \emph{Localization at weak disorder: some elementary bounds}, Rev. Math. Phys. \textbf{6}, 1163--1182 (1994).

\bibitem[AM]{AM} M. Aizenman and S. Molchanov: \emph{Localization at large disorder and at extreme energies: an elementary derivation}, Commun. Math. Phys. \textbf{157}, 245--278 (1993).

\bibitem[ASW]{ASW} M. Aizenman, R. Sims and S. Warzel: \emph{Stability of the Absolutely Continuous Spectrum of Random Schr\"odinger Operators on Tree Graphs}, Prob. Theor. Rel. Fields \textbf{136}, no. 3, 363--394 (2006). 

\bibitem[AW1]{AW1} M. Aizenman and S. Warzel: \emph{Absence of mobility edge for the Anderson random potential on tree graphs at weak disorder}, Euro. Phys. Lett. \textbf{96}, 37004 (2011)

%\bibitem[AW2]{AW2} M. Aizenman and S. Warzel: \emph{Extended States in a Lifshitz Tail Regime for Random Schr\"odinger Operators on Trees}, Phys. Rev. Lett. \textbf{106}, 136804 (2011).

\bibitem[AW2]{AW2} M. Aizenman and S. Warzel: \emph{Resonant delocalization for random Schrödinger operators on tree graphs}, J. Eur. Math. Soc. \textbf{15}(4), 1167--1222 (2013).

%\bibitem[B]{B} J. Breuer: \emph{Localization for the Anderson Model on Trees with Finite Dimensions}, Ann. Henri Poincar\'e \textbf{8}, 1507--1520 (2007).

\bibitem[CKM]{CKM}  R. Carmona, A. Klein and F. Martinelli: \emph{Anderson localization for Bernoulli and other singular potentials}, Commun. Math. Phys. {\bf 108}, 41-66 (1987)

\bibitem[DLS]{DLS}  F. Delyon, Y. Levy and B. Souillard,  \emph{Anderson localization for multidimensional systems at large disorder or low energy},  Commun. Math. Phys. {\bf 100}, 463-470 (1985)

\bibitem[DK]{DK}  H. von Dreifus and A. Klein: \emph{A new proof of localization in the Anderson tight binding model},  Commun. Math. Phys. {\bf 124}, 285-299 (1989)

%\bibitem[DSS]{DSS} F. Delyon, B. Simon, and B. Souillard: \emph{From power pure point to continuous spectrum in disordered systems}, Ann. Inst. H. Poincar\'e Phys. Th\'eor. \textbf{42}, no. 3, 283--309 (1985).

\bibitem[FHH]{FHH} R. Froese, F. Halasan and D. Hasler: \emph{Absolutely continuous spectrum for the Anderson model on a product of a tree with a finite graph}, J. Funct. Analysis \textbf{262}, 1011-1042 (2012)

\bibitem[FHS1]{FHS1} R. Froese, D. Hasler and W. Spitzer: \emph{Transfer matrices, hyperbolic geometry and absolutely continuous spectrum for some discrete Schr\"odinger operators on graphs}, J. Funct. Anal. \textbf{230}, 184--221 (2006).

\bibitem[FHS2]{FHS1b} R. Froese, D. Hasler, and W. Spitzer: \emph{Absolutely continuous spectrum for the Anderson model on a tree: geometric proof of Klein's theorem}, Commun. Math. Phys. \textbf{269}, 239--257 (2007).

\bibitem[FHS3]{FHS2} R. Froese, D. Hasler and W. Spitzer: \emph{Absolutely continuous spectrum for random potentials on a tree with strong transverse correlations and large weighted loops}, Rev. Math. Phys. \textbf{21}, 1--25 (2009).

\bibitem[FMSS]{FMSS}  J. Fr\"ohlich, F. Martinelli, E. Scoppola and T. Spencer: \emph{Constructive proof of localization in the Anderson tight binding model},  Commun. Math. Phys. {\bf 101}, 21-46 (1985)

\bibitem[FS]{FS} J. Fr\"ohlich and T. Spencer: \emph{Absence of diffusion in the Anderson tight binding model for large disorder or low energy}, Commun. Math. Phys. \textbf{88}, 151--184 (1983).

\bibitem[GMP]{GMP}  Ya. Gol'dsheid, S. Molchanov and L. Pastur: \emph{Pure point spectrum of stochastic one dimensional Schr\"odinger operators}, Funct. Anal. Appl. {\bf 11}, 1-10 (1977)

\bibitem[Ha]{Hal} F. Halasan: \emph{Absolutely continuous spectrum for the Anderson model on trees}, PhD thesis 2009, arXiv:0810.2516v3 (2008)

\bibitem[HSS]{Hamzaetal} E. Hamza, R. Sims and G. Stolz: \emph{A note on fractional moments for the one-dimensional continuum Anderson model}, J. Math. Anal. Appl. {\bf 365}, 435-446 (2010)

\bibitem[Hu]{H} D. Hundertmark: \emph{A short introduction to Anderson localization}, Analysis and stochastics of growth processes and interface models, 194--218, Oxford Univ. Press, Oxford, 2008. 

%\bibitem[HP]{HP} P. Hislop and O. Post: \emph{Exponential localization for radial random quantum trees}, Waves Random Complex Media, \textbf{19}(2), 216--261 (2009).

%\bibitem[Kir]{Kir} W. Kirsch: \emph{An Invitation to Random Schr\"odinger operators}, Soc. Math. France, Panoramas \& Synth\`esis, no 25, 1--119 (2008).

%\bibitem[KLW]{KLW} M. Keller, D. Lenz and S. Warzel: \emph{On the spectral theory of trees with finite cone type}, Israel J. Math. \textbf{194}, 107-135 (2013) 

\bibitem[KLW]{KLW} M. Keller, D. Lenz and S. Warzel: \emph{Absolutely continuous spectrum for random operators on trees of finite cone type}, J. d'Analyse Math. \textbf{118} 363-396 (2012)

\bibitem[K]{K} A. Klein: \emph{Extended States in the Anderson Model on the Bethe Lattice}, Adv. in Math. \textbf{133}, 163--184 (1998). 

\bibitem[KLS]{KLS} A. Klein, J. Lacroix, and A Speis: \emph{Localization for the Anderson Model on a Strip with Singular Potentials}, J. Funct. Anal. \textbf{94}, 135--155 (1990).

\bibitem[KlS]{KlS} A. Klein and C. Sadel: \emph{Absolutely Continuous Spectrum for Random Schr\"odinger Operators on the Bethe Strip}, Mathem. Nach. \textbf{285}, 5--26 (2012).

\bibitem[Klo]{Klo} F. Klopp: \emph{Weak disorder localization and Lifshitz tails}, Commun. Math. Phys. \textbf{232}, 125-155 (2002)

\bibitem[KuS]{KS} H. Kunz and B. Souillard: \emph{Sur le spectre des des op\'erateurs aux diff\'erences finies al\'eatoires}, Commun. Math. Phys. \textbf{78}, 201--246 (1980).

\bibitem[La]{Lac}  J. Lacroix: \emph{Localisation pour l'op\'erateur de Schr\"odinger al\'eatoire dans un ruban},  Ann. Inst. H. Poincar\'e ser {\bf A40}, 97-116 (1984)

\bibitem[Lee]{L} D. Lee: \emph{Anderson localization in one dimension}, Undergraduate research project (Math 448) at UBC (2013).

\bibitem[Sa]{Sa} C. Sadel: \emph{Absolutely continuous spectrum for random Schr\"odinger operators on tree-strips of finite cone type}, Annales Henri Poincare, {\bf 14}, 737-773 (2013)

\bibitem[SS]{SS} C. Sadel and H. Schulz-Baldes: \emph{Random Dirac Operators with time reversal symmetry}, Commun. Math. Phys. {\bf 295}, 209-242 (2010)

\bibitem[Sch]{S} J. Schenker: \emph{Eigenvector localization for random band matrices with power law band width}, Commun. Math. Phys. \textbf{290}, 1065--1097 (2009).

\bibitem[Sh]{Sh} M. Shamis, \emph{Resonant delocalization on the Bethe strip}, preprint, arXiv:1304.4461

\bibitem[SW]{SW} B. Simon and T. Wolff: \emph{Singular Continuous Spectrum under Rank One Perturbations and Localization for Random Hamiltonians}, Commun. Pure and Appl. Math. \textbf{34}, 75--90 (1986).

\bibitem[Sto]{Sto} G. Stolz: \emph{An introduction to the mathematics of Anderson localization},
Entropy and the quantum II, 71--108, Contemp. Math. \textbf{552}, Amer. Math. Soc., Providence, RI, 2011.

%\bibitem[Sto2]{Sto2} G. Stolz: \emph{Strategies in localization proofs for one-dimensional random Schr\"odinger operators}, Proc. Indian Acad. Sci. (Math. Sci.) \textbf{112}, 229-–243 (2002).

\bibitem[Stu]{Stu} K.T. Sturm: \emph{Probability measures on metric spaces of nonpositive curvature}, Heat kernels and analysis on manifolds, graphs, and metric spaces (Paris, 2002), 357--390, Contemp. Math. \textbf{338}, Amer. Math. Soc., Providence, RI, 2003. 

\bibitem[W]{W} W.-M Wang: \emph{Localization and universality of Poisson statistics for the multidimensional Anderson model at weak disorder}, Invent. Math. \textbf{146}, 365-398 (2001)

\end{thebibliography}
\end{document}